\documentclass[reqno,keywordsasfootnote,10pt]{article}
\usepackage[english]{babel}
\usepackage[latin1]{inputenc}
\usepackage{amsmath,amsthm,amsfonts,amssymb,times}
\numberwithin{equation}{section}
\usepackage[final]{epsfig}
\usepackage{color}

  \usepackage{epstopdf}
\newtheorem{theorem}{Theorem}[section]

\newtheorem{lemma}[theorem]{Lemma}
\newtheorem{corollary}[theorem]{Corollary}
\theoremstyle{definition}
\newtheorem{definition}{Definition}[section]

\newcommand{\E}{\mathbb{E}}
\newcommand{\Z}{ {\mathbb Z} }
\newcommand{\N} {{\mathbb N}}
\newcommand{\C}{\mathbb{C}}
\newcommand{\Co}{\mathcal{C}}

\newcommand{\R}{\mathbb{R}}

\DeclareMathOperator{\supp}{supp}
\DeclareMathOperator{\dist}{dist}
\DeclareMathOperator{\sgn}{sign}

\DeclareMathOperator{\diam}{diam}

\newcommand{\bx}{\mathbf{x}}
\newcommand{\by}{\mathbf{y}}

\newcommand{\bv}{\mathbf{v}}
\newcommand{\bw}{\mathbf{w}}

\newcommand{\balpha}{\boldsymbol{\alpha}}

\def\E#1{\mathbb{E} \left[ #1 \right]}  
\def\g[#1,#2,#3]{\langle #1|\frac{1}{#2} |#3\rangle} 
  
\makeatletter

\DeclareMathOperator{\indfct}{\rm 1}

\makeatother
\def\be{\begin{equation} }
\def\ee{\end{equation} }
\begin{document}
\date{\small September 18, 2008 \\ 
                     Rev. Dec. 22, 2008}

\title{Localization Bounds for Multiparticle Systems}   

\author{Michael Aizenman and Simone Warzel\footnote{Present address: Zentrum Mathematik, TU M\"unchen}\\ \small Departments of
Mathematics and Physics,\\ \small Princeton University, Princeton NJ
08544, USA.}%


\maketitle

\begin{abstract}

We consider the spectral and dynamical properties of quantum systems of $n$ particles on the lattice $\Z^d$, of arbitrary dimension, with a Hamiltonian which  in addition to the kinetic term  includes  a random potential with iid values at the lattice sites and  a finite-range interaction.    Two basic parameters of the model are  the strength of the disorder and the strength of the interparticle interaction.  
It is established here that for all $n$ there are regimes of high disorder, and/or weak enough interactions, for which the  system exhibits spectral and dynamical localization.    
The localization is expressed through bounds on the transition amplitudes,  which are uniform in time and decay exponentially in the Hausdorff distance in the configuration space. The results are derived  through the analysis of fractional moments of the $n$-particle Green function, and related bounds on the eigenfunction correlators.
\\[5ex]
(2000 Mathematics Subject Classifiction: 47B80, 60K40)\\
\vspace*{1cm}
\end{abstract}

\newpage

\tableofcontents

\newpage

\newpage
\section{Introduction}  

\subsection{On localization in the presence of interactions}

In the context of non-interacting particle systems, or equivalently one-particle theory,  Anderson localization is a well studied phenomenon, which for various regimes of the parameter space can be established even at the level of rigorous mathematical analysis~(e.g.~\cite{CaLa90,PF,St00,ASFH,GK2,Ki07s} and references therein).  The picture is far less complete when it comes to systems of interacting particles subject to a random external potential, which generally may be expected to produce localization.   
 
Particularly perplexing is the situation where there are $n$ fermions in a region of volume $|\Lambda|$, with $|\Lambda| \to \infty$ and  $n/|\Lambda| \to \rho >0$.  It was proposed, through reasoning presented in ~\cite{BAA06}, that if the interactions are week and the mean particle separation is significantly below the localization length of the non-interacting system, then the interaction would not affect by much the
dynamical properties of the system.  

   In particular, such reasoning has lead to the suggestion that if the system is started in a configuration for which the density of particles  in one part  of the region is higher that in another then the uneven situation will persists indefinitely, assuming the Hamiltonian is time independent.   While that would be in violation of the  equipartition principle, it would be in line with the dynamical behavior of the the non-interacting system  in the regime of complete Anderson localization for the one-particle Hamiltonian (\cite{A2,ASFH,GK1}).   

Rigorous methods  are still  far from allowing one to decide whether complete localization will persist in the presence of interactions, as claimed in ~\cite{BAA06}.  
Furthermore,  the analysis of even  a fixed number of particles with short range interactions, and $|\Lambda|\to \infty$, has presented difficulties.   An important step was recently made by Chulaevsky and Suhov~\cite{ChSu1,ChSu3}  who proved the existence of spectral localization for systems of two interacting particles which are subject to highly disordered external potential.   The authors expect that the analysis of the $n=2$ case, which is based on the multiscale approach of~\cite{FrSp83,vDK89}, could be extended to any finite $n$.    In this work we approach the question using  somewhat different tools, and address also the issue of dynamical localization.   
We establish the existence of localization regimes for any finite $N$, with decay rates which are uniform in the volume.  
Curiously,  as is  indicated in the figure below,  
the bounds which are established here carry a qualitatively somewhat stronger implication for $n=2$ than for higher values of $n$.  

\subsection{Statement of the main result} 

Our goal here is to present a basic proof of localization for an arbitrary number of particles moving on a lattice of arbitrary dimension, which for convenience is taken to be $\Z^d$, in  regimes of high disorder or sufficiently weak interactions.

There are a number of ways to formulate a quantum system of particles on a lattice, which here is taken to be $\Z^d$.  We find the following convenient, but the method discussed here can be also be adapted to other formulations. 

The Hilbert space of $ n $ particles on  $\Z^d$ is the direct product ${\mathcal H}^{(n)} =  \ell^2(\mathbb{Z}^d)^n$. We take the  Hamiltonian to be an  operator of the form:  
\be \label{Ham} 
H^{(n)}(\omega) \ = \ \sum_{j=1}^n   [-\Delta_j  + \lambda \, V(x_j; \omega) ] \ +\   {\mathcal U}(\bx;\balpha)  \, , 
\ee 
with: $ \Delta $  the discrete Laplacian (second difference operator) in $\Z^d$, $V(x,\omega)$  a random potential (described below), and ${\mathcal U}$ a finite range interaction   which is given in terms of functions of the occupation numbers
\be 
{\mathcal U}(\bx;\balpha)  \ =\ \sum_{k=2}^{p}  \alpha_k   \sum_{
\substack{ 
A\subset \Z^d:  |A|=k \\ 
{\rm diam}A \le \ell_U
}  }   U_A( {\mathcal N}_A(\bx))  
\ee 
where 
\be 
{\mathcal N}_A(\bx) \ = \ \{N_{u}(\bx)\}_{u\in A} \, , 
\ee 
with $N_{u}(\bx) = \sum_{j=1}^n \delta_{x_j, u}$  the number of particles the configuration $ \bx \in (\Z^d)^n $ 
has at $u\in \Z^d$.  It is to be understood that $U_A (\bx  )=0$ unless $\prod_{u\in A} N_u(\bx) \neq 0$.   
    
The family of Hamiltonians   is parametrized by  $ \lambda \in \mathbb{R}_+ $, which controls the strength of the disorder, and  $ \balpha := (\alpha_2 , \dots , \alpha_p) \in \mathbb{R}^{p-1} $ which is the strength of the interaction.  Obviously,  the value of $\alpha_k$  is of relevance for $H^{(n)}(\omega)$ only for $ n \geq k $.   
It will be assumed throughout the paper that: 
\begin{itemize}
	\item[{\bf A1}] The random potential is given in terms of a collection of iid random variables, $ \left\{V(x; \omega)\right\}_{x \in \mathbb{Z}^d} $, with 
 \be \mathbb{E}\left[ \, \exp( t\,  |V(0)| )\, \right] < \infty \,  \qquad \mbox{ for all $t\in \R$}, 
\ee
whose probability distribution is of bounded density, i.e.,  
\be 
  \mathbb{P}(V(x) \in d\xi ) = \varrho(\xi) \, d\xi \    \qquad \mbox{with} \quad \varrho \in L^\infty  \, , 
  \ee
satisfying:
\be \label{s_reg}
 \varrho(v) \  \leq \  K \int_{|v'|\le E_0} \varrho(v- u)\,  du\,  , \qquad \mbox{ for all $u\in \R$}, 
 \ee
at some  $E_o<\infty $ and $K<\infty$.  
\item [{\bf A2}]  
	The interaction terms are bounded, with $ |U_A (\boldsymbol{n}  )|\, \le \ 1$ for all $A\subset \Z^d$ and all $\boldsymbol{n} \in (\Z^d)^{|A|} $, and translational invariant, i.e., $ U_A = U_{A'} $ if $ A' $ is a translate of $ A$.
\end{itemize} 	

The above assumptions could be relaxed.  In particular, 
translation invariance can be replaced by suitable translation invariant  bounds, and, as in the case of one particle localization theory,  the absolute continuity of the measure  and \eqref{s_reg} can be replaced by a local power-law concentration bound such as the following condition: 
\be \label{s_reg2}
 \mathbb{P}( V(x) \in [v+\varepsilon,v-\varepsilon   ] )  \  \leq \  \varepsilon^\tau   K_\tau \,   \mathbb{P}( V(x)\in [v+E_0,v-E_0   ] ) \ .
 \ee
for some  $\tau \in [0,1)$ and all $0<\varepsilon \le 1$.  
  Under such reduced assumption,   for which the case $\tau =1$ corresponds to \eqref{s_reg}, the fractional moment bounds presented below would be limited to $0 < s < \tau$ (and  minor adjustment will be required in the argument, cf. \cite{ASFH}), but that would not adversely affect the main results.    \\

Our main result is naturally  stated in terms of the eigenfunction correlator  which is introduced in Section~\ref{sec:corr}.  However, the statement can also be presented as follows. 
\begin{theorem} \label{thm:main}  
Under the assumptions {\bf A1} - {\bf A2} , for each $p$ and $ n \in \N$ there is an open set in the parameter space, $\Gamma^{(p)}_n \subset \R_+\times \R^{p-1}$ such that:
\begin{enumerate} 
\item For all $(\lambda, {\bf \balpha}) \in \Gamma^{(p)}_n$ and up to $ n $ particles, i.e., $ k \in \{ 1, \dots , n \} $, each operator $H^{(k)}(\omega)$ has almost surely only pure point spectrum, with the corresponding eigenfuntions being exponentially localized in the sense of  $\rm{dist}_{\mathcal H}(\bx, \bx_0)$, as is explained below.  
\item 
Furthermore, for all $(\lambda, {\bf \balpha}) \in \Gamma^{(p)}_n$, all $ k \in \{ 1, \dots , n \} $, and all $\bx,\, \by \in (\Z^d)^k$:
\begin{equation} \label{loc_f}
\E{\sup_{  \|f\|_\infty \le 1}\left|\langle \delta^{(k)}_\bx \, ,  \, f(H^{(k)}) \, \delta^{(k)} _\by \rangle \right|}   \ \le \  A \, 
e^{-  \rm{dist}_{\mathcal H}(\bx,\by) / \xi} \,  ,
\end{equation} 
where 
\begin{equation} \label{eq:distH}
\dist_{\mathcal H}(\bx, \by) 
\ := \ \max\left\{ \max_{1\le i \le k} \dist(\{x_i\},\, Y), \   \max_{1\le i \le k} \dist(\{y_i\},\, X)
\right\} \, ,  
\end{equation} 
with $\delta_{\bx}^{(k)}\, \big[\delta_{\by}^{(k)}\big]$  the  $k$-particle position eigenstates of corresponding to $\bx \, [\by]$,  and the constants $A$ and $\xi$ depending on $ n, p$  but not on $k \, (\le n)$. 
\item The localization region 
$\Gamma^{(p)}_n $ includes regimes of strong disorder and of weak interactions, i.e.,
   \begin{enumerate} 
   \item      for each $\balpha \in \R^{p-1}$ there is $\lambda(\balpha )$ such that   $\Gamma^{(p)}_n  \supset (\lambda(\balpha), \infty) \times \{\balpha \}$.
  \item for any $\lambda \in \Gamma^{(1)}_1$, i.e. one for which the one-particle Hamiltonian exhibits complete localization, there are  $\balpha(\lambda)_j > 0$
 , $j=\{1,...,p\}$, 
  such that $\Gamma^{(p)}_n$ includes all $(\lambda, \balpha') $ for which 
  $|\balpha'| \le  |\balpha (\lambda)|$ componentwise.  
    \end{enumerate} 
  \end{enumerate} 
\end{theorem} 

The bound \eqref{loc_f}, applied to  $f(H)=e^{-itH}$, implies  dynamical localization, and through that also the spectral assertion which is made in Theorem~\ref{thm:main}.  The latter is explained more explicitly in Appendix~\ref{app:spec}.   

One may note that the  distance between configuration, $\dist_{\mathcal H}(\bx,\by)$, which appears above  corresponds to the Hausdorff distance between the sets $ X=\bigcup_{i=1}^{n} \{x_i\}$ and $Y=\bigcup_{i=1}^{n} \{y_i\}$, seen as subsets of $\Z^d$  with its Euclidean metric.  
The  exponential bound presented above deserves a number of further comments.   

\subsection{Remarks on the rate of exponential decay}

For systems of non-interacting particles, i.e., the case $\balpha ={\bf 0}$,  the one-particle localization theory allows 
to prove that  in regimes of localization  a bound like 
 \eqref{loc_f} holds with a stronger decay rate.  For the stronger bound (which can be established for regimes of strong enough disorder, or extremal energies, and in one dimension the full range of energies)   
the relevant distance is not   $\rm{dist}_H(\bx,\by) $ but:
\begin{equation} 
\rm{dist}(\bx,\by) \ := \  \sum_{j=1}^n |x_j-y_j| \, . 
\end{equation}   

One could ask whether exponential decay in  $\rm{dist}(\bx,\by)$  persists also in the presence of interactions.    
Upon reflection, the general answer to this should be negative:
The  $n$ particle  Hamiltonian \eqref{Ham} clearly commutes with elements of the permutation group $S_n$.   In the non-interacting case (i.e.,  $\balpha ={\bf 0}$)  its spectrum is degenerate, 
$H^{(n)}(\omega)$ being the sum of $n$ commuting unitarily equivalent operators, each affecting only one particle. 
However, since generically the interactions couple the different permutation-related degenerate eigenstates of the non-interacting systems, it is natural to expect that for  $\balpha \not\equiv{\bf 0}$, 
the operator $H^{(n)}(\omega)$  will have no  eigenstates in which the probability amplitude is essentially restricted to the vicinity of a particular $n$ particle configuration.  
Instead, localization may still be manifested in the existence of  eigenstates which decay in the sense of the symmetrized configuration distance 
\begin{equation} 
\rm{dist}_S(\bx,\by) \ := \ \min_{\pi \in S_n}  \sum_{j=1}^n |x_j-y_{\pi j}| \, , 
\end{equation} 
with $S_n$ is the permutation group of the  $n$ elements $\{1, ...,n\}$.  

The dynamical version of this eigenfunction picture is that for very   large $t$   a  state  of the form $e^{-itH^{(n)}} \delta^{(n)}_{\bx}$, which has evolved from the initially localized state at $\bx$,  would have non-negligible amplitude not only in the vicinity of $\bx$ but also in the vicinity of the permuted configurations $\pi \bx $.     
 
The above considerations are of course superfluous in case one is interested only in the fully symmetric or antisymmetric sector, where the initial states cannot be localized in the stronger sense and where only decay rates which are symmetric under permutations are of relevance.  
However,  the decay  rate $\exp(-\dist_{\mathcal{H}}(\bx,\by)/\xi)$ is still qualitatively weaker than  $\exp(-\rm{dist}_S(\bx,\by)/\xi')$.  In particular, for configurations which include some tight subclusters  with multiple occupancy our bounds do not rule out the possibility, which we do not expect to be realized, of the excess charges being able to hop freely between the different subclusters, as depicted in Fig. 1.   Nevertheless, the bounds allow to conclude the main features of localization.   Some further explicit comparison between different distances are given in  Appendix~\ref{app:dist}. \\

\begin{figure}[h]\label{fig1}
\begin{center}
\includegraphics[width=.8\textwidth] {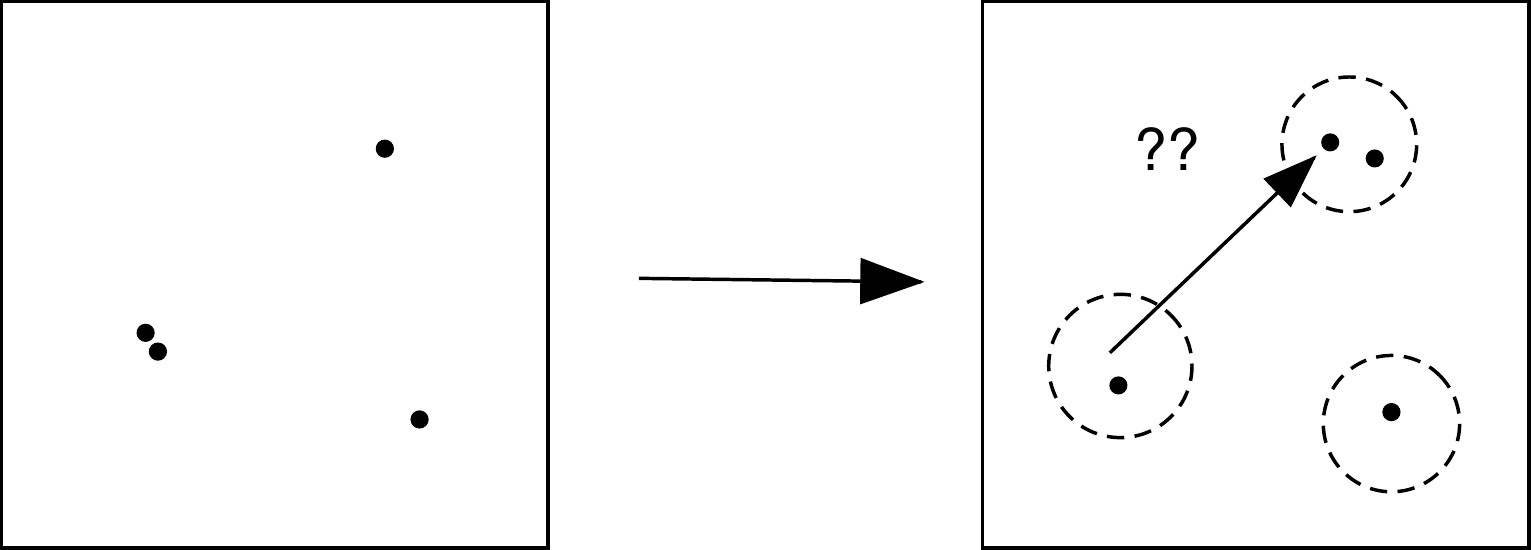}
\caption{A schematic depiction of our localization bounds: Starting from the configuration depicted on the left, at any later time except for events of very small probability  the collection of locations which are near an occupied site  does not change by much.  However,  the Hausdorff metric bounds allow for the possibility that if the initial configuration had two or more particle in sufficient proximity [within the localization distance]  then such `excess' may transfer among the occupied regions.}  
\end{center}
\end{figure}

\subsection{Comments on the proof} 

It would be natural to ask why is there a need for a separate proof of localization for the $n$-particle system, since the configurations of the system can also be viewed as describing a single particle with $nd$-dimensional position vector, $\bx = (x_1, ..., x_n)$.  Regarded from this perspective, the Hamiltonian \eqref{Ham} may at first appear to have the usual structure for which localization is well understood, consisting of the usual  kinetic term, a potential function $U(\bx)$, and a random potential $ \sum_j  V (x_j; \omega)$.  
The answer is that the values which the random potential $V(\bx; \omega) $ assumes at different positions  in the $nd$-dimensional space are not independent.  Instead, they are correlated over arbitrary distances, and the number of its independent degrees of freedom ($L^d$), for the systems in a box of linear size $L$,  scales as only a fractional power of the number of its configurations ($L^{nd}$).   From this perspective,   the randomness is much more limited than what is found in the well understood one-particle situation.    

The  proof of Theorem~\ref{thm:main} is organized as induction on $n$, and  is guided by the following picture:  once localization is established for less than $n$-particles, one may expect that throughout most of the volume the time evolution of $n$ particles does not disperse, except possibly when the particles are all close to each other and move as some $n$-particle cloud.   This hypothetical mode, with the $n$ particles particles forming a quasi-particle, is ruled out using one-particle techniques, which are modified to show that such a  possibility does not occur at weak enough interactions.   

Technically, our proof makes use of the Green function fractional-moment techniques, and in particular the  finite-volume criteria of \cite{ASFH}.   However, in addition to adapting a number of ``off the shelf'' one-particle arguments we need to show that exponential decay of the fractional moment of the Green function for lower numbers of particles implies:  {\em i.\/}  exponential decay for systems composed of non-interacting subsystems, and  {\em ii.\/}  for the interactive system  - exponential decay  in a distance defined relative to the set of clustered $n$-particle configurations.   These terms are explained more explicitly in the following sections. 

In Sections~\ref{sec:BGF}--\ref{sec:ILSS} we present a number of relations which  are utilized in the derivation of Theorem~\ref{thm:main}.   These are then strung to a proof in Section~\ref{sec:proof}.

\section{Finiteness of the Green function's fractional moments}\label{sec:BGF}

The  proof of localization proceeds through exponentially decaying estimates for  the Green function $ G_\Omega(\bx,\by;z) $ of finite volume versions of   $H^{(n)}(\omega)$ evaluated at energies $ z $ within the spectrum of the infinite volume operator. 
Our first step is to establish that for $s<1$ each $ |G_\Omega(\bx,\by;z)|^s $  is of finite   conditional expectation value, regardless of $z\in \C$,  when averaged over one or two potential variables -- provided each of the configurations  ($\bx$ and $\by$)  includes at least one  averaged site.   
The basic strategy is familiar from the theory of one-particle localization.  However the proofs need to be revisited here since   we are now dealing with random potentials whose values for different configurations are no longer independent, and in certain ways are highly correlated.    

It may be noted that  the celebrated Wegner estimate is not being explicitly used in the Fractional Moment Analysis, its role being taken by the  finiteness of the Green function's fractional moments.   However, in view of the Wegner's estimate's  intrinsic interest and conceptual appeal and we comment on it in Appendix~\ref{Sect:Wegner}.   \\

We shall use the following notation: 
The $ n $-particle Green function associated with some region $ \Omega \subseteq \mathbb{Z}^d $ and $ z \in \mathbb{C}^+ $, is 
\begin{equation}
	G_\Omega(\bx,\by;z) \equiv G_\Omega^{(n)}(\bx,\by;z) := \left\langle \delta^{(n)}_\bx , \left( H^{(n)}_\Omega - z \right)^{-1} \delta^{(n)}_\by \right\rangle \, . 
\end{equation}
where 
$ H^{(n)}_\Omega(\omega) $ is  the (finite-volume) operator  obtained by resticting \eqref{Ham} to the Hilbert space 
$ \mathcal{H}^{(n)}_\Omega :=  \ell^2(\Omega)^n $  (with the default choice of  Dirichlet boundary conditions).
The vectors $ \delta^{(n)}_\bx , \delta^{(n)}_\by \in  \mathcal{H}^{(n)}_\Omega $ correspond to localized states, i.e. $ \langle \delta^{(n)}_\bx , \psi \rangle = \psi(\bx) $, and are parametrized by 
configurations $ \bx = (x_1, \dots, x_n ) $ of $ n $-particles.  When clear from the context we shall drop the superscript $ (n) $ at our convenience.

The set of  configurations with all particles in $ \Omega $ is denoted by $ \mathcal{C}^{(n)}(\Omega) := \Omega^n $.
Also:
\begin{itemize}
\item
For a given set $ S \subset \Omega $, we denote by $ \mathcal{C}^{(n)}(\Omega;S) $ the set of $ n $-particle configurations which have  at least one particle in $ S $. 
In case $ S = \{x\} $ the set will also be denoted as $ \mathcal{C}^{(n)}(\Omega;x) $.

\item  We denote by $ \mathcal{C}^{(n)}_r(\Omega) := \{ \bx \in \Omega^n \, | \, \diam( \bx ) \leq r \} $ 
 the set of configurations with diameter less or equal to $ r $, 
the diameter of a configuration being defined as $ \diam(\bx) \ := \ \max_{j,k}\,   |x_j - x_k|  $. 
\end{itemize}

\begin{theorem}\label{prop:frac}
For any $ s \in (0,1) $ there exists $ C_s < \infty $ such that for   
any $ \Omega  \subseteq \mathbb{Z}^{d}$, any two (not necessarily distinct) sites $ u_1,u_2 \in \Omega $ and any pair of configurations  
$ \bx \in\mathcal{C}^{(n)}(\Omega;u_1) $ and $ \by \in  \mathcal{C}^{(n)}(\Omega;u_2) $, the following bound holds
\begin{equation}\label{eq:frac}
	\mathbb{E}\left( \left| G^{(n)}_\Omega(\bx,\by;z)\right|^s \, \Big| \, \left\{ V(v) \right\}_{v\not\in \{u_1,u_2\}} \right) \leq C_s \, \frac{  (K E_0)^\# }{(|\lambda| E_0)^s}  
\end{equation}
for all $ z \in \mathbb{C}^+ $ and $ \lambda \neq 0 $, with $\#=2$ in case $u_1 \neq u_2$, and $\#=1$ otherwise. 
\end{theorem}
For the proof, let us note that in its dependence on the single-site random variables $V(x)$ the Hamiltonian is of the form
\begin{equation}\label{eq:HN}
	H_\Omega(\omega)=  A + \lambda \sum_{u \in \Omega } V(u; \omega) \, N_u \, ,
\end{equation}
where $ N_u $ is the number operator, $ \left(N_u \psi \right)(\bx) := \sum_{k=1}^n \delta_{x_k,u} \, \psi(\bx)  $, which 
counts the number of particles on the site $ u \in \mathbb{Z}^d $. 

In analyzing averages over the potential  variables we shall employ the following \emph {double sampling bound}, which is the dual form of the regularity assumption {\bf A1}, Eq.~\eqref{s_reg}.  
Under that assumption, for  any non-negative function $h$ of one of the single potential parameters $V\equiv V(u)$, for some $u\in \Z^d$:
\be \label{2_av}
\int_\R h(V) \,\rho(V) \, d V \ \le \ K E_0 \int_\R \int_{|V'|\le E_0 } h(V + V') \  \frac{dV'}{E_0} \ \rho(V)\,  d V \,  \, .   
\ee

\begin{proof}[Proof of Theorem~\ref{prop:frac}]
We shall estimate the conditional expectation in \eqref{eq:frac} with the help of the double sampling bound \eqref{2_av} , applied to the pair of variables $V(u_j)$  (or a single one in case they coincide).  
According to that, it suffices to estimate just the integral over the variable(s) $V'(u_j)$ 
of the fractional moment of 
\be \label{wideG}
 {G}'_\Omega(\bx,\by;z)  \ := \  \langle \delta_{\bx}\, , \big(\, H +   \sum_{w \in \{u_1,u_2\}} V'(w) N_w 
	- z \big)^{-1} \delta_{\by} \rangle \, ,
\ee

At this point, a useful tool is the following weak $L^1 $-estimate which forms a straightforward extension of \cite[Prop.~3.2]{AENSS}: 
For any pair of normalized vectors $ \phi , \psi $ in some Hilbert space, any pair of self adjoint operators with $ N, M \ge 0 $, and a maximally dissipative operator $ K $: 
\begin{equation}\label{eq:AENSS}
	\int_{[-1,1]^2}  \indfct \left[  \left| \left\langle \phi , 
	\sqrt{N} \left( \xi \, N+ \eta \, M   -  K \right)^{-1} 
	\sqrt{M}  
	\psi \right\rangle \right| > t \right]   \, d\xi \, d\eta 
		\leq \frac{C}{t}  \, ,
\end{equation}
for all $ t > 0 $ with some (universal) constant $ C < \infty $, where $ \indfct[\dots] $ denotes the indicator function.  A similar bound holds for the one-variable version of ~\eqref{eq:AENSS}.

Applying \eqref{eq:AENSS}   to the expression in \eqref{wideG}, and noting that $ \delta_\bx $ and $ \delta_\by $ are eigenvectors of $ N_{u_1} $ and $ N_{u_2}$ with eigenvalues greater or equal to one, one gets:
\begin{eqnarray}
W(t) &:=& E_0^{-2}\left| \left \{ \big( V'(u_1),V'(u_2)\big) \in [-E_0,E_0]^2   \ :  
\left| 
     {G}'_\Omega(\bx,\by;z)  
      \right|  \ge t \right \}  \right|  \  \nonumber \\ 
	  &\le &  \min \{ 4, \frac{C }{|\lambda|  E_0 \  t}  \}
\end{eqnarray}
where $W(t)$ is  introduced just for the next formula.   To estimate the corresponding integral of  the kernel's fractional moment, one may use the  Stieltjes integral expression:   
\begin{eqnarray} 
	\int_{[-E_0,E_0]^2}     
	      \left| {G}'_\Omega(\bx,\by;z)  \right|^s  \frac{dV'(u_1)}{E_0} \frac{ dV'(u_2)}{E_0}	
	& = &  \int_0^\infty 
	W(t)  \  d(t^s)      \nonumber \\ 
	& \leq & \frac{4^{1-s}  C^s} {(1-s) \, (|\lambda| E_0)^s } \,      .  
\end{eqnarray}
A similar bound holds for in case $ u_1 = u_2 $ for the average over a single variable. 
The bound \eqref{eq:frac}  is implied now through a simple application of~\eqref{2_av}.
\end{proof}

\section{Localization domains in the parameter space}  

The following notions are useful in describing localization bounds which persist when the strength of the disorder is driven up, and also to present the localization regimes which are discussed in this work.   
For their formulation we denote by  $ \lambda_1 \in \R_+$  the critical coupling above 
which the one-particle Hamiltonian $ H^{(1)} $ exhibits uniform $ 1 $-particle localization in the sense of Definition~\ref{def:loc} below.
Its existence was established in \cite{AM}. 
\begin{definition} \label{def:robust}.
A {\em robust domain} in the parameter space is a non-empty open set $\Gamma \subset \R_+\times \R^{p-1}$ for some $ p \in \mathbb{N} $ such that: 
\begin{enumerate} 
\item if $(\lambda,\balpha) \in \Gamma$, then for all
 $\lambda' > \lambda$ also $(\lambda', \balpha) \in \Gamma$
\item for every $\balpha \in \R^{p-1}$  there exist
$\lambda(\balpha) \in \R_+$ such that 
$(\lambda(\balpha), \balpha)\in \Gamma$, 
\item  $\Gamma$ includes the half-line
$\big(\lambda_1,\infty\big)\times \{ \boldsymbol{0}\} $.
\end{enumerate}
\end{definition}  

\begin{definition}\label{def:subcon}  
A subset of the parameter space $\Gamma \subset \R_+ \times \R^{p-1}$  is called  {\em sub-conical}, if  
there is some $ c < \infty $ such that for all $ (\lambda,\balpha) \in\Gamma$
\be 
 \sum_{k=2}^p   \frac{(2\ell_U)^{d k}}{k!} |\alpha_k|  
 \  \leq \ c \, |\lambda|  
\ee \end{definition} 
In this context
it is worth noting that under Assumption~{\bf A2} the interaction obeys the  bound: 
\begin{equation}\label{eq:Hbound}
\big\| \mathcal{U}(\balpha) \big\|_\infty := \sup_{\bx\in (\Z^d)^n}  |{\mathcal U}(\bx;\balpha)| \ \le \  n \ \left[ \sum_{k=2}^p   \frac{(2\ell_U)^{d k}}{k!} |\alpha_k| \right] \, .
   \end{equation}

A useful criterion of  localization is expressed in terms of the  fractional moments of the Green function, 
with the average being carried out over both the disorder  and  the energy within an interval $ I \subset \mathbb{R} $.  For this purpose we denote:  
\begin{equation}
	\widehat{\mathbb{E}}_I[\cdot] := |I|^{-1} \, \int_I \mathbb{E}\left[ \cdot \right] dE \, .
\end{equation}

\begin{definition} \label{def:loc}
A robust subset of the parameter space, $\Gamma \subset \R_+\times \R^{p-1}$ is said to be a \emph{domain of uniform $n$-particle localization}
if for some $s\in (0,1)$ there exists $\xi = \xi(s,n,p) < \infty$ and $ A = A(s,n,p) < \infty $ such that for all $ (\lambda,\balpha) \in \Gamma $, all $k\in \{1,...,n\}$, and all  $\bx, \by \in \mathcal{C}^{(k)}$:
\begin{equation} 
\label{eq:exp}
	\sup_{\substack{I \subset \mathbb{R} \\ |I| \geq 1}} \sup_{\Omega \subset \mathbb{Z}^d}  \widehat{\mathbb{E}}_I\left( \left| G^{(k)}_\Omega(\bx,\by)\right|^s    \right) \ \leq \ 
	A\,  e^{-\dist_\mathcal{H}(\bx,\by)/\xi} \, ,  
\end{equation}
where the energy variable on which $G$ depend was averaged over the intervals $I$. 
\end{definition}
In the above definition we have incorporated the specific choice of the distance, $ \dist_{\mathcal H} $, only for convenience.  
As was explained in the introduction, it seems natural to expect exponential decay also in terms of  the symmetrized distance, $ \dist_S $, but that is not proven  here.      

Let us also add that localization in the sense of \eqref{eq:exp} implies   various other, more intuitive and physically relevant, expressions of the phenomenon; in particular dynamical (Section~\ref{sec:corr}) as well as spectral localization (Appendix~\ref{app:spec}).

In this work we shall focus on the proof of existence of  robust regimes of localization for any finite $n$ and $p$, without monitoring closely the values of the localization length $\xi$, and amplitude $A$.  
In particular, the subsequent proof yields a localization length which degrades heavily when the number of particles $ n $ increases.

Concerning the value of $ s $ in the above definition, it is helpful to notice
\begin{lemma}
If $ |\lambda | $ is bounded away from zero and the condition \eqref{eq:exp}  is satisfied for some $s\in (0,1)$ 
then it holds for all other $s\in (0,1)$ at  adjusted values of $\xi < \infty$ and $A < \infty$.
\end{lemma}
\begin{proof}
Jensen's and H\"older's inequality imply that for all $ r \leq s  \leq t < 1 $
\begin{multline}\label{eq:all4one}
\left(\widehat{\mathbb{E}}_I\left[ \left| G_\Omega(\bx,\by)\right|^r  \right] \right)^{\frac{s}{r}} \leq \widehat{\mathbb{E}}_I\left[ \left| G_\Omega(\bx,\by)\right|^s \right] \\
\leq \left(\widehat{\mathbb{E}}_I\left[ \left| G_\Omega(\bx,\by)\right|^t \right] \right)^{\frac{s-r}{t-r}} \left(\widehat{\mathbb{E}}_I\left[ \left| G_\Omega(\bx,\by)\right|^r \right] \right)^{\frac{t-s}{t-r}} \, .
\end{multline}
The first term in the last line is bounded, $\widehat{\mathbb{E}}_I\left[ \left| G_\Omega(\bx,\by)\right|^t \right]  \leq C(t) \, |\lambda|^{-t} $, thanks to~\eqref{eq:frac}. 
\end{proof}

\section{Multiparticle eigenfunction correlators and the Green function} \label{sec:corr}

 \subsection{Eigenfunction correlators} 

A convenient expression of localization, and also a convenient tool for the analysis, is provided by the 
eigenfunction correlators. 
By this term we refer to   the family of kernels, 
for $\bx, \by \in (\Z^d)^n$:  
 \begin{equation}
 	Q^{(n)}_{\Omega}(\bx,\by;I;s) := \mkern-20mu 
	\sum_{E  \in \sigma(H^{(n)}_{\Omega})\cap  I }  \mkern-25mu \langle \delta^{(n)}_\bx \, ,  P_{\{E\} }(H^{(n)}_\Omega) \, \delta^{(n)}_\bx \rangle^{1-s}  \,
	 \left| \langle \delta^{(n)}_\bx \, ,  P_{\{E\} }(H^{(n)}_\Omega) \, \delta^{(n)}_\by \rangle \right|^{s} \, ,
 \end{equation}
 where $I\subset \R$ is a subset of the energy range,  $\Omega \subset \Z^d$ is a finite subset, $  P_{\{E\} }(H_\Omega^{(n)}) $ is 
 the spectral projection on the eigenspace corresponding to the eigenvalue  $E$, and $s\in [0,1]$ is an interpolation parameter.   
 This  definition extends naturally to also  unbounded $\Omega\subset \Z^d$, provided $ H^{(n)}_{\Omega} $ has only pure point spectrum within $I$.   
 The notation used here differs from that of \cite{A2} by allowing for degeneracies in the spectrum.  As mentioned in the introduction, 
 while the spectrum of a one-particle Hamiltonian with random potential is almost surely non-degenerate, 
 degeneracies do occur in the non-interacting multiparticle case.
 
When the domain $\Omega$ and the value of $n$ are clear from the context, or of no particularly importance, the sub/super-scripts   on $Q$ may be suppressed.  When $s$ is omitted, it is understood to take the value $s=1$, which for many purposes is the most relevant one.    

 An essential property of the kernel is the bound (at $s=1$):  
\begin{equation}  \label{Q-F}
\sup_{ \|f\|_\infty \le 1} |\langle \delta_\bx \, , \, f(H_{\Omega}) \,  \delta_\by \rangle  |  \ \le \  \left| Q_{\Omega}(\bx,\by;I)\right| \, .
\end{equation}

In its dependence on the parameter $s$ the kernel is  log-convex, i.e., for any $ \lambda \in [0,1] $ 
\begin{equation} \label{Q2}
	 Q(\bx,\by;I; (1-\lambda ) p_0 +  \lambda p_1) \leq  
	 	Q(\bx,\by;I;s_0)^{(1-\lambda )} \, Q(\bx,\by; I;s_1)^{\lambda}\, . 
\end{equation}
Moreover: 
\begin{align} \label{Q1}
Q(\bx,\by;I;0) \  & =  \sum_{E  \in \sigma(H_{\Omega})\cap  I }  
 \langle \delta_\bx \, ,  P_{\{E\} }(H_\Omega) \, \delta_\bx \rangle \ \leq 1 \ \, , \nonumber \\  
 Q(\bx,\by;I;1) \  & \le  \sum_{E  \in \sigma(H_{\Omega})\cap  I }  
 \left|  \langle \delta_\bx \, ,  P_{\{E\} }(H_\Omega) \, \delta_\by \rangle  \right| \ \leq 1\, .   
\end{align}
where the latter is by the Schwarz inequality.  A useful implication of equations (\ref{Q1}) and (\ref{Q2}) is that for any $ 0 < s < t \leq 1 $ 
\begin{equation}\label{eq:increase}
Q(\bx,\by;I; t) \leq Q(\bx,\by;I; s)^ {\frac{t-s}{1-s}}    \, .
\end{equation}

The relations \eqref{Q1} and \eqref{Q2}  played a role in the strategy which was used in \cite{A2} for the deduction of dynamical localization through Green function fractional moment bounds.   As we shall see next, the method can be extended to many particle systems.   

Most of our analysis will be done in finite volumes.  A minor subtlety concerning the passage to the infinite volume limit, is that we do not have an a-priori statement of   convergence in this limit of the eigenfunctions, nor of the eigenfunction correlators.  Nevertheless, one has the following statement.   

\begin{theorem}   \label{thm:corr->spec}
Suppose that the following bound holds for a sequence of finite domains $\Omega$ which converge to $\Z^d$, and a fixed interval $I \subset \R$, 
\begin{equation} \label{loc_assump}
\E{ Q_{\Omega}^{(n)}(\bx,\by;I)}   \ \le \  A \; 
e^{-  K (\bx,\by)  } \,  ,
\end{equation} 
 with 
$K (\cdot,\cdot)$ some kernel, i.e., a two point function defined over the space of pairs of 
$n$-particle configurations, and some $A = A(n)  < \infty$.
Then, within the $n$-particle sector, the infinite volume operator  $H(\omega)$ satisfies: 
\begin{equation} \label{loc_fK}
\E{\sup_{ \|f\|_\infty \le 1}\left|\langle \delta^{(n)}_\bx \, ,  \, f(H^{(n)}) \, \delta^{(n)}_\by \rangle \right|}   \ \le \  A \, 
e^{-  K(\bx,\by) } \,  .
\end{equation} 
Furthermore, if  \eqref{loc_assump} holds with $K 
(\bx,\by) =2\rm{dist}_{\mathcal H}(\bx,\by)/\xi$  then one may also conclude that the $n$-particle spectral projection on $I$ is almost surely given by a sum over a collection of  rank-one projections on eigenstates which decay exponentially, each satisfying a bound of the form: 
\be  \label{loc_bound}
|\psi(\bx;\omega)|^2 \ \le \  A(\omega;n) \, \left( 1 + |\bx_\psi|\right)^{2nd+2} \, e^{-\dist_{\mathcal{H}}(\bx,\bx_\psi)/\xi} \, ,
\ee
where $A(\omega;n) $ is an  amplitude of finite mean, and the decay is from a configuration $\bx_\psi$ at which the wave function is non-negligible in the sense that 
\be \label{eq:nonneg}
|\psi(\bx_\psi;\omega)|^2 \ \geq \ \frac{\left( 1 + |\bx_\psi|\right)^{-(nd+1)} }{
	\sum_{\by\in (\Z^d)^n} \left( 1 + |\by|\right)^{-(nd+1)}}  \, .  
\ee   
\end{theorem} 
With the natural modification, the last statement is valid also in case $K(\bx,\by) $ is given in terms of  any of the other distances  which were mentioned in the introduction, i.e., $\rm{dist}(\bx,\by)$ or $\rm{dist}_S(\bx,\by)$. 

Except for a  minor reformulation of a known bound, this relation  is in essence well familiar from the theory of one particle localization (it was used already in \cite{A2}).  We therefore  relegate its proof to the Appendix~(\ref{app:spec}). 

As it turns out,  averages over the disorder of the eigenfuction correlator are closely related with  Green function's fractional moments.  The rest of this section is devoted to the relations between the two quantities.

\subsection{Lower bound in terms of Green function's  fractional moments}

The following (deterministic) estimate allows to bound fractional moments of Green functions in terms of eigenfunction correlators. 

\begin{theorem}\label{lem:G<Q}
Let $ \Omega \subset \mathbb{Z}^d $. For any $ s \in (0,1) $ and any interval $ I \subset \mathbb{R} $
\begin{equation}\label{eq:G<Q}
	\int_I \left| G^{(n)}_{\Omega}(\bx,\by;E) \right|^s \, dE \leq  \frac{2 \; |I|^{1-s} }{1-s} \; Q^{(n)}_\Omega(\bx,\by,\mathbb{R})^s \, .
\end{equation}
\end{theorem} 
One may note that this bound is useful only in case of complete localization of all eigenfunctions, but that suffices for our purpose.   
The bound may be improved with a restriction of the  eigenfunction correlator to a finite,  slightly enlarged, interval $ I' \supset I $; the contribution to the Green function 
from eigenfunctions outside $ I$  being handled with the help of quasi-analytic cutoff in the sense of  Helffer-Sj\"ostrand, and the Combes-Thomas Green function estimate.  

\begin{proof}[Proof of Theorem~\ref{lem:G<Q}]
We split the contribution to the Green function into two terms depending on whether $ \langle \delta_\bx \, ,  P_{\{E\} }(H_\Omega) \, \delta_\by \rangle \geq 0 $ 
or $  \langle \delta_\bx \, ,  P_{\{E\} }(H_\Omega) \, \delta_\by \rangle < 0 $,
\begin{equation}
	G^{\pm}(\bx,\by;z) := \mkern-40mu \sum_{\substack{E \in \sigma(H) \\ \sgn \langle \delta_\bx \, ,  P_{\{E\} }(H_\Omega) \, \delta_\by \rangle = \pm} }\mkern-30mu \frac{ \langle \delta_\bx \, ,  P_{\{E\} }(H_\Omega) \, \delta_\by \rangle}{E - z} 
\end{equation}
(Note that the eigenfunctions of \eqref{Ham} may be taken to be real. In the complex case one would have four terms instead.)
Using $|a+b|^s \leq |a|^s + |b|^s $  we thus get
\begin{align}\label{eq:splitGpm}
\int_I \left| G(\bx,\by;E) \right|^s \, dE \leq & \sum_{\#=\pm} \int_I \left|G^\#(\bx,\by;E)\right|^s \, dE  \notag \\
 =  &  \sum_{\#=\pm} s \int_0^\infty \left| \left\{ E \in I \, \big| \, \left|G^\#(\bx,\by;E)\right| > t \right\} \right| \, \frac{dt}{t^{1-s}} 
\end{align}
Boole's remarkable formula, which states that $ \left| \left\{ x \in \mathbb{R} \, \big| \, \big|\sum_n p_n (x_n - x )^{-1} \big| > t \right\} \right|= 2 \sum_n p_n \, t^{-1} $ for all $ x_n\in  \mathbb{R}$,  $p_n , t > 0,\, $ \cite{Boo57}, implies that
\begin{equation}\label{eq:Boole}
\left| \left\{ E \in \mathbb{R} \, \big| \, \left|G^\#(\bx,\by;z)\right| > t \right\} \right| = \frac{2}{t} \mkern-40mu \sum_{\substack{E \in \sigma(H_\Omega) \\ \sgn  \langle \delta_\bx  ,  P_{\{E\} }(H_\Omega)  \delta_\by \rangle = \#} }\mkern-50mu  \langle \delta_\bx \, ,  P_{\{E\} }(H_\Omega) \, \delta_\by \rangle =: \frac{2}{t} Q^\#(\bx,\by,\mathbb{R})  \, . 
\end{equation}
Substituting in integral   \eqref{eq:splitGpm} the maximum of \eqref{eq:Boole} and  the length of the interval, $|I|$,
one  arrives at
\begin{eqnarray}
\int_I \left| G(\bx,\by;E) \right|^s \, dE    & \leq  & \frac{ 2^s |I|^{1-s}}{1-s}  \left[ Q^+(\bx,\by,\mathbb{R})^s + Q^-(\bx,\by,\mathbb{R})^s \right] \ \nonumber \\[1ex]  
 & \leq &  \frac{2\, |I|^{1-s} }{1-s}\, Q(\bx,\by,\mathbb{R})^s \, . 
\hfill 
 \end{eqnarray}
\end{proof}

\subsection{Upper bound in terms of Green function's  fractional moments}

For the proof of our main result we need also a converse bound to~\eqref{eq:G<Q}. 

In the one-particle case there is a simple passage from exponential decay of  Green function fractional moments to similar bounds on the mean value of the eigenfunction correlators, and thus to dynamical localization~\cite{A2}.   In effect, it is based on the following relation, which is a somewhat more explicit statement than what is found in~\cite{A2}.   
 
\begin{lemma}  \label{lem:QF1} For any finite domain $\Omega \subset \Z^d$,  $x\in \Lambda$,  
$ s \in (0,1) $ and Borel set $ I \subset \R$,  
\begin{equation}\label{eq:repFM}
 \int_\mathbb{R}     Q^{(1)}_\Omega(x,y;I,s)\big|_{V_x \to V_x + v}  \,  \frac{dv}{|v|^s} \ = \  |\lambda|^{s-1} \int_I \left|\langle \delta_x , (H^{(1)}_\Omega-E)^{-1} \delta_y \rangle\right|^s   \, dE 
\end{equation}
where the left side involves the eigenfunction correlator for the one parameter family of operators $ H^{(1)}_\Omega(v)  :=  H^{(1)}_\Omega  +\lambda \, v\, P_x $ acting in $\ell^2(\Lambda)$. 
\end{lemma}

The combination of \eqref{eq:repFM} and the `double sampling bound' \eqref{2_av}  produces the desired upper bound on the the expectation of the eigenfunction correlator - in the one-particle case.   

To extend this argument to the multiparticle case, an extension is needed of the averaging principle which is expressed in Lemma~\ref{lem:QF1}.   Following is a suitable generalization.

\begin{lemma}\label{lem:specavQ}
Let $ s \in [0,1) $ and  $ \Omega \subset \mathbb{Z}^d $ be a finite set, and  $u $ a point in $ \Omega $. Then for all $ \bx \in \mathcal{C}^{(n)}(\Omega;u) $ and $ \by \in \mathcal{C}^{(n)}(\Omega) $:
\begin{multline} \label{refresh}
	N_u(\bx) ^{1+\frac{s}{2}}\ \, \int_\mathbb{R} Q_\Omega(\bx,\by;I,s) \Big|_{V(u)\to V(u)+v}\, \frac{dv}{|v|^s} \\
	= \ \frac{1}{|\lambda|^{1-s}} \int_I \,  \sum_{\kappa\in \sigma(K_{u}(E))}
	\langle \delta_\bx ,  \Pi_\kappa (E)   \,  \delta_\bx\rangle^{1-s} \,  \\ \times  \left|\langle \delta_\bx  ,  \Pi_\kappa (E) \sqrt{N_u} \, (  H_\Omega  - E)^{-1}   \delta_\by \rangle\right|^s \, dE \, . 
\end{multline} 
where  $ \Pi_\kappa (E)  \equiv  P_{\{\kappa\}}(K_u (E) )$ 
is the spectral projection on the eigenspace at eigenvalue $\kappa$  for the $E$- dependent operator: 
\begin{equation}
	K_u(E) \  :=\  \sqrt{N_u} (  H_\Omega  -E)^{-1} \!\sqrt{N_u}   \, ,  
\end{equation}
which we take as acting  within the range of $  N_u $ in 
$\mathcal{H}^{(n)}_\Omega $.  
\end{lemma}

Since the proof takes one on a technical detour, we have placed it here in the Appendix~\ref{app:averaging}.   Using this averaging principle,  we get: 

\begin{theorem}\label{thm:Q<G}
Let  $\Omega \subset \Z^d$ be a finite subset, and $ u \in \Omega $. Suppose $\bx, \by \in \Co^{(n)}(\Omega)$ is a pair of configurations such that the number $ N_u(\bx) $ of particles of $ \bx $ at $ u $ is at least one. 
Then for any  $ s \in (0,1) $ and any interval $ I \subset \mathbb{R} $
\begin{multline}
N_u(\bx) \; \mathbb{E}\left[Q_\Omega^{(n)}(\bx,\by;I,s)\right] \\
	\leq \frac{K \, |E_0|^s}{|\lambda|^{1-s} } \sum_{\bw \in \mathcal{C}^{(n)}(\Omega;u)} \left(\frac{N_u(\bw)}{N_u(\bx)}\right)^{s/2} \int_I \E{\left|G^{(n)}_\Omega(\by,\bw;E)\right|^s}  dE \, .
\end{multline}
\end{theorem}
\begin{proof}
It follows from  \eqref{2_av} that the conditional expectation of any non-negative function $f$ of the random variables $\{V(x)\}_{x\in \Omega}$, conditioned on the values of $V$ at sites other than $u$, satisfies:
\begin{equation} \label{R5} 
	\mathbb{E}\left[ f(V(u)) \, \big| \, \{V(x)\}_{x\neq u} \right] \leq K \, |E_0|^s \;  \mathbb{E}\left[ \int f(V(u) + v)  \frac{dv}{|v|^s}  \, \Big| \, \{V(x)\}_{x\neq u} \right]
\end{equation}

We apply this relation to $f$  the eigenfunction correlator.  The quantity  which one then finds on the right side of \eqref{R5} can be rewritten with the help of Lemma~\ref{lem:specavQ}.   The claimed bound  then easily follows using  
the fact that  $ |a+b|^s \leq |a|^s + |b|^s $ (for $ 0< s < 1$), and   $ \sum_{\kappa} \langle \delta_\bx ,   \Pi_\kappa(E)   \,  \delta_\bx\rangle^{1-s} \,
\left|\langle \delta_\bx  ,  \Pi_\kappa(E) \delta_\bw \rangle\right|^s \leq 1 $, in \eqref{refresh}.
\end{proof}

Applications of the above result are restricted to bounded intervals $ I $. In order to control the eigenfunction correlator associated with the tails of the spectrum we also use:  

\begin{lemma}\label{lem:tailbounds}
Let $ \Omega \subseteq \mathbb{Z}^d $ and $ E \geq 0 $. Then for every $ \bx \in \mathcal{C}^{(n)}(\Omega) $:
\begin{multline}\label{eq:tailbounds}
	\mathbb{E}\left[\left\langle \delta^{(n)}_\bx \, , \, P_{\R\backslash (-E,E)}\big(H^{(n)}_\Omega\big) \, \delta^{(n)}_\bx \right\rangle\right] \\
		\leq  \mathbb{E}\left[e^{|V(0)|} \right]  \, 
		\exp\left\{   \min\{1,(n |\lambda|)^{-1} \}  \big(2dn + \big\| \mathcal{U}(\balpha) \big\|_\infty - E\big)\right\} \, .
\end{multline}
\end{lemma}
\begin{proof}
The Chebyshev-type inequality,  
$	 \indfct_{\R\backslash (-E,E)}(x)\leq e^{-tE} \left( e^{tx} + e^{-tx} \right)$, reduces the bound to one on the semigroup
for which we employ  the Feynman-Kac representation (cf.\ \cite[Prop.~II.3.12]{CaLa90}) to show that for any $ t>0 $:
\begin{align}
	& \mathbb{E}\left[ \left\langle \delta_\bx \, , e^{t H_\Omega }\delta_\bx \right\rangle\right]\\
		& \leq \int \mathbb{E}\left[\exp\left(\int_0^t \left(\sum_{u\in\Omega}\lambda V(u)N_u(\by(s)) + \mathcal{U}(\by(s);\balpha) \right) ds  \right)\right] \,   \nu^{(\bx;t)}_\Omega(d\by) \notag \\
		& \leq \mathbb{E}\left[ e^{ n t  \, |\lambda| |V(0)|}\right] \, \exp\left\{  t \big(2dn +  \big\| \mathcal{U}(\balpha) \big\|_\infty \big)\right\} 
		\,  \left\langle \delta_\bx \, , e^{ t \sum_j\Delta_{\Omega,j} }\delta_\bx \right\rangle \, , \notag
\end{align}
where $ \nu^{(\bx;t)}_\Omega $ is the measure generated by $ \sum_j  \Delta_{\Omega,j} $ on path $ \{\by(s) \}_{0\leq s \leq t} $, starting in and returning to $ \bx $ in time $ t $. 
The last inequality is a version of Jensen applied to average $ (n t )^{-1} \sum_{u \in \Omega} \int_0^t ( \cdot) \, N_u(\by(s)) \, ds $ in the exponential. 
A similar bound holds for $ t < 0 $. 
The proof is completed using the operator bound $	 \left\langle \delta_\bx \, , e^{t \Delta_\Omega }\delta_\bx \right\rangle \leq e^{2dn |t|} $
for the free semigroup and the choice $ t = \min\{1,(n |\lambda|)^{-1} \}$.
\end{proof}


\section{Implications of localization in subsystems}\label{sec:ILSS} 

In the induction step of the proof of the main result, we shall be considering  for a system of $n$ particles the consequences of localization bounds which are already established for subsystems.   In this section we present some results which will be useful for that purpose; first considering the case when the subsystems are combined without interaction, and then the more involved situation where the two subsystems are coupled via short range interaction.  
For a partition of the index set $\{1,...,n\}$ into disjoint subsets $J$ and $K$, 
we denote the coordinates of the two subsystems as $\bx_J := \{x_j\}_{j\in J}$ and correspondingly $\bx_K$.    

\subsection{Localization for non-interacting systems}

When two subsystems are put together with no interaction, the Hamiltonian is --  in natural notation,
 \be\label{eq:compsys}
  H^{(J,K)}_\Omega := H^{(J)}_\Omega \oplus H^{(K)}_\Omega\, ,
   \ee
   acting in $ \ell^2(\Omega)^{|J|} \oplus \ell^2(\Omega)^{|K|} $.
A complete  set of eigenfunctions of  the operator sum  can be obtained by taking 
products of the subsystems' eigenfunctions.
Clearly, if the subsystems  exhibited spectral localization, that property will be inherited by the composite system. 

The question of localization properties of the corresponding Green function, which is an important tool for our analysis, is a bit less immediate: $ G^{(J,K)}_\Omega $ is a convolution, with respect to the energy, of the subsystems' Green functions, i.e., 
for any $ \bx = (\bx_J,\bx_K) $, $  \by = (\by_J,\by_K) $, and  $ z \in \mathbb{C}\backslash \mathbb{R} $: 
\begin{equation}
	 G^{(J,K)}_\Omega(\bx,\by;z) = \int_\mathcal{C} G^{(J)}_\Omega(\bx_J,\by_J;z-\zeta) \, G^{(K)}_\Omega(\bx_K,\by_K;\zeta) \; \frac{d\zeta}{2\pi i } \, ,
\end{equation}
where
$ \mathcal{C} $ is any closed contour in $ \mathbb{C} $, 
which encloses the spectrum of $H^{(K)}_\Omega $ but none of  $ H^{(J)}_\Omega - z $.
Given the singular nature of the $E$-dependence of the Green function, localization in the sense of Definition~\ref{def:loc} is  not immediately obvious.  
To establish it, we  take a  detour via eigenfunction correlators.  These  are less singular in $E$, but share 
the convolution structure, which  in this case can be written in the form: 
\begin{align}\label{eq:Qfactor}
	Q^{(J,K)}_\Omega(\bx,\by;I) = & \int Q^{(J)}_\Omega(\bx_J,\by_J; I -E ) \, Q^{(K)}_\Omega(\bx_K,\by_K; dE ) 	\notag  \\
		\leq &\,  Q^{(J)}_\Omega(\bx_J,\by_J; \mathbb{R}) \; Q^{(K)}_\Omega(\bx_K,\by_K; \mathbb{R}) \, , 
\end{align}

The following result will allow us to apply this relation.  

\begin{lemma}\label{cor:efcor}
Let $ \Gamma \in \R_+\times \R^{p-1} $ be a sub-conical domain of uniform $ n $-particle localization.
Then there exist  $ A , \xi \in (0,\infty) $ such that the eigenfunction correlator corresponding to up to $ n $ particles, i.e., $k \in \{ 1 , \dots, n \}$, is exponentially bounded for all $ (\lambda,\balpha)\in \Gamma $, and all $ \bx , \by \in(\Z^d)^k$:
\begin{equation} \label{eq:Qbound}  
	\sup_{\Omega \subset \mathbb{Z}^d}  \E{Q_\Omega^{(k)}(\bx,\by;\mathbb{R})} \leq \; A \, e^{- \dist_{\mathcal H}(\bx,\by)/\xi} \, .
\end{equation}
\end{lemma}
\begin{proof} 
As an immediate consequence of  \eqref{eq:increase}, Theorem~\ref{thm:Q<G}, and Lemma~\ref{lem:sum}, we know that there is $  A , \xi \in (0,\infty) $ 
such that 
\begin{equation}\label{eq:efcorproof}
 \E{Q_\Omega^{(k)}(\bx,\by;[ -E , E])} \leq \;  A \, \frac{ E}{|\lambda|} \, e^{ -\dist_{\mathcal{H}}(\bx,\by)/\xi} \, , 
\end{equation}
for any  $ E > 0 $ and any $ (\lambda,\balpha) \in \Gamma $.  
For a bound which is uniform in $E$ we combine this with 
Lemma~\ref{lem:tailbounds}, which with the help of  the Cauchy-Schwarz inequality implies:
\begin{align}
	& \E{Q^{(k)}_\Omega(\bx,\by;\mathbb{R})} - \E{Q_\Omega^{(k)}(\bx,\by;[ -E , E])} \notag \\ 
	&	\leq  \left( \E{Q^{(k)}_\Omega(\bx,\bx;\mathbb{R}\backslash [-E,E])} \,\E{Q^{(k)}_\Omega(\by,\by;\mathbb{R}\backslash [-E,E])} \right)^{1/2} \notag \\
	& \leq    \mathbb{E}\left[e^{|V(0)|} \right]  \, \exp\left\{ \min\{1,( k |\lambda|)^{-1}\} \left( 2d k + \big\| \mathcal{U}(\balpha) \big\|_\infty - E \right) \right\}  \, . 
\end{align}
Choosing the cutoff at $ E = 2d k + \big\| \mathcal{U}(\balpha) \big\|_\infty +  \max\{1,k \, |\lambda|\} \dist_{\mathcal H}(\bx,\by)  / \xi $ one obtains the claimed exponential bound.  
Using the fact that  $ \Gamma $ is sub-conical in sense of Definition~\ref{def:subcon}, the above argument yields a $(\lambda,\balpha) $-independent  
amplitude ``$A$'' in~\eqref{eq:Qbound}.
\end{proof}

The two-way relation between the eigenfunction correlators and fractional moments of the Green function, and the factorization property \eqref{eq:Qfactor},  allows us now to establish:

\begin{theorem}\label{thm:noninteracting} 
Let $ \Gamma \in \R_+\times \R^{p-1} $ be a sub-conical domain of uniform $ n $-particle localization. For any $s\in (0,1)$ there is $\xi ,A \in (0, \infty) $ such that  the  Green function of the composition~\eqref{eq:compsys} of any pair $ (J,K) $ of systems of  at most $n $ particles (i.e., $\max\{ |J|, |K| \} \leq n $) is bounded for all $ (\lambda,\balpha) \in \Gamma $ and all $ \bx = (\bx_J,\bx_K) $, $  \by = (\by_J,\by_K) $:
\begin{equation} 
	\sup_{\substack{I \subset \mathbb{R} \\ |I| \geq 1}} \sup_{\Omega \subset \mathbb{Z}^d}  \widehat{\mathbb{E}}_I\left( \left| G^{(J,K)}_\Omega(\bx,\by)\right|^s    \right) \ \leq \ 
	A\,  e^{-\dist_{\mathcal H}^{(J,K)}(\bx,\by)/\xi} \, ,  
\end{equation}
with 
$  \dist_{H}^{(J,K)}(\bx,\by) := \max\{\dist_{\mathcal H}(\bx_J, \by_J) ,\dist_{\mathcal H}(\bx_K, \by_K) \}$. \end{theorem} 
\begin{proof}
By Theorem~\ref{lem:G<Q}, and the Jensen's inequality,  for any $ s \in (0,1) $ there is a constant $ C = C(s) < \infty $ such that 
\begin{equation}
	\widehat{\mathbb{E}}_I\left[ \left| G_\Omega^{(J,K))}(\bx,\by)\right|^{s}  \right] \leq  \frac{C}{| I |^s} \, \E{Q^{(J,K)}_\Omega(\bx,\by;I)}^s \, .
\end{equation}
The claim follows by combining:  {\em i)\/}  the product formula \eqref{eq:Qfactor},  {\em ii)\/} the uniform bound $ Q \leq 1 $, and   {\em iii)\/}  the bound of Lemma~\ref{cor:efcor}, applied to the factor with the greater separation.
\end{proof}

\subsection{Decay away from clustered configurations}

We now turn to the more involved situation, where a system consists of subsystems which separately exhibit localization, but which  are put together with an interaction. 
Intuitively,  the decay of the fractional moments for the subsystems should imply smallness of the corresponding kernel  for the composite system for pairs of configurations where  at least one of the pair can be split into two well separated parts.  

To express this idea in a bound, we shall use the  notion of the splitting width of a configuration:
\begin{equation}
 \ell(\bx):= \max_{\substack{J,K \\  J\dot{\cup} K = \{1 , \dots , n \}}}  \quad \min_{j\in J , \, k \in K } | x_j - x_k| \, , 
 \end{equation} 
 where the maximum runs over all the two-set partitions of the index set $ \{1 , \dots , n \}$. It is easy see that $ \diam(\bx)/ (n-1) \leq \ell(\bx)  \leq \diam(\bx) $ (cf.\ Appendix~\ref{app:dist}).

 \begin{theorem}\label{thm:apriori}
Let $ \Gamma^{(p)}_{n-1} \subset \R_+ \times \R^{p-1} $  be  a sub-conical  domain of uniform $ (n-1) $-particle localization, with $ n \geq 2 $.   
Then there exist some 
 $ s \in (0,1) $, $ A , \xi < \infty $ such that for all $ (\lambda,\balpha) \in \Gamma^{(p)}_{n-1} $ and all $ \bx, \by \in (\Z^d)^n $:
\begin{multline}\label{eq:apriori}
	\sup_{\substack{I \subset \mathbb{R} \\ |I| \geq 1}} \sup_{\Omega \subseteq \mathbb{Z}^d }\, 
	\widehat{\mathbb{E}}_I\left[ |G_\Omega^{(n)}(\bx,\by)|^{s}\right]\\
	 	 \leq \; A  \, \exp\left( - \frac{1}{\xi} \min\left\{ 
		 	\dist_{\mathcal H}(\bx,\by), \max\{\ell(\bx), \ell(\by)\} \right\} \right)  \, .
\end{multline}
\end{theorem}

\begin{proof} 
We fix $ \bx  , \by \in \mathcal{C}^{(n)}(\Omega)$ and assume without loss of generality that $ \ell(\bx) \geq \ell(\by) $. 
We then split $ \bx $ into two clusters,  $\bx_{J}, \bx_{K} $  such  that
\begin{equation}\label{eq:mindistD}
	\ell(\bx) = \min_{j\in J , \, k \in K } | x_j - x_k|  \, .
\end{equation}
Between the two clusters, $\bx_{J}, \bx_{K} $ we remove all interactions, leaving the inter-cluster interaction untouched. 
The resulting operator is a direct sum of two non-interacting subsystems, $ H^{(J,K)}_\Omega :=  H^{(J)}_\Omega \oplus H^{(K)}_\Omega $ 
acting in $ \ell^2(\Omega)^{|J|} \otimes \ell^2(\Omega)^{|K|} $
where 
\begin{equation}
 H^{(J)}  := \sum_{j=1}^{|J|}  -\Delta_j  + \lambda \, V_\omega(x_j) 
	 	+ \sum_{k=2}^{|J|} \alpha_k \mkern-5mu\sum_{\substack{A \subset \mathbb{Z}^d \, : \, |A| =  k \\ \diam A \leq \ell_U}}    U_A(\mathcal{N}_A(\bx_J))
\end{equation} 
and similarly for $ H^{(K)} $. Denoting the Green function corresponding to $ H^{(J,K)}_\Omega $ by $ G^{(J,K)}_\Omega$,
and using $ |a + b |^s \leq |a |^s + |b|^s $ and we thus have
\begin{multline}\label{eq:resolventNI1}
\widehat{\mathbb{E}}_I\left[ |G_\Omega(\bx,\by)|^{{s}}\right]  \leq 
		\widehat{\mathbb{E}}_I\left[ |G^{(J,K)}_\Omega(\bx,\by)|^{{s}}\right] 		 + 	\widehat{\mathbb{E}}_I\left[ |G^{(J,K)}_\Omega(\bx,\by)-G_\Omega(\bx,\by)|^{{s}}\right] 
\end{multline}
Since $ G^{(J,K)}_\Omega$ is a Green function of a composite system, whose parts are assumed to exhibit uniform $ (n-1) $-particle localization, 
Lemma~\ref{thm:noninteracting} 
guarantees the existence of $ {s} \in (0,1)  $ and $ A , \xi \in (0,\infty) $ such that for all $ (\lambda , \balpha) \in   \Gamma^{(p)}_{n-1} $ and all $ \bx , \by $: 
\be \label{eq:expboundNI}
\widehat{\mathbb{E}}_I\left[ \left|G^{(J,K)}_\Omega(\bx,\by)\right|^{s}\right] 
	\  \leq  \  A \, e^{ - \dist^{(J,K)}_{\mathcal H}(\bx,\by)/\xi}   \ \leq  \  A \, e^{ - \dist_{\mathcal H}(\bx,\by)/\xi}  \, ,
\ee
where the last step is by the general relation
$ \dist_{\mathcal H}^{(J,K)}(\bx,\by)   \   \geq \  \dist_{\mathcal H}(\bx,\by) $.
To bound the second term we use the resolvent identity 
\begin{equation}\label{eq:resolventNI}
	 \Delta := \; G^{(J,K)}_\Omega(\bx,\by;z) -  \,  G_\Omega(\bx,\by;z)   
	 = \mkern-5mu \sum_{\bw \in \mathcal{C}^{(n)}(\Omega)}\!\!\! G^{(J,K)}_\Omega(\bx,\bw;z) \, U_{J,K}(\bw) \, G_\Omega(\bw,\by;z) \, ,
	 \end{equation}
where $ U_{J,K}  = H - H^{(J,K)} $. In order to be able to apply the Cauchy-Schwarz inequality we first decrease the exponent with the help of H\"older's inequality, 
\begin{equation}
	 \widehat{\mathbb{E}}_I[ |\Delta|^{s} ] \leq \left(\widehat{\mathbb{E}}_I[ |\Delta|^{\frac{3{s}(1+{s})}{2(1+2{s})}} ]\right)^{\frac{1+2{s}}{2+{s}}} 
	 \left( \widehat{\mathbb{E}}_I[ |\Delta|^{\frac{{s}}{2}} ]\right)^{\frac{1-{s}}{2+{s}}} \leq c\;  |\lambda|^{-\frac{3{s}(1+{s})}{2(2+{s})} }
	 \left( \widehat{\mathbb{E}}_I[ |\Delta|^{\frac{{s}}{2}} ]\right)^{\frac{1-{s}}{2+{s}}} \, , 
\end{equation}
where the last inequality is due to \eqref{eq:frac}. 
Inserting \eqref{eq:resolventNI} and using the Cauchy-Schwarz inequality together with \eqref{eq:frac} we thus obtain 
\begin{equation}
	\widehat{\mathbb{E}}_I[ |\Delta|^{s} ] \leq \frac{c}{|\lambda|^{s}}  \sum_{\bw \in \mathcal{C}^{(n)}(\Omega)} \left|  U_{J,K}(\bw) \right|^{{s}\beta} \, \left(\mathbb{E}\left[\left|G^{(J,K)}_\Omega(\bx,\bw;z)\right|^{{s}} \right]\right)^{\beta} 
\end{equation}
where we abbreviated $ \beta := \frac{1-{s}}{2(2+{s})} $.
To estimate the right side note that   
\begin{align}
 \sup_{\bx} | U_{J,K}(\bx)|  = &
 \sup_\bx \Big|\sum_{k=2}^{|J|} \alpha_k \mkern-5mu\sum_{\substack{A \subset \mathbb{Z}^d \, : \, |A| =  k \\ \diam A \leq \ell_U}}    U_A(\mathcal{N}_A(\bx)) \notag \\
 	& \qquad \times \indfct\big[ \mbox{There is $ j \in J $}, \mbox{$ k \in K $ s.t.} \; | x_j - x_k | \leq \ell_U \big] \, \Big| \notag \\
 		 \leq & \; \big\| \mathcal{U}(\balpha) \big\|_\infty \, .
\end{align}
Moreover, the distance of $ \bx $ to  the support of 
$U_{J,K} $
is bounded from below,
\begin{equation}
	 \inf_{\bw \in  \supp U_{J,K} } \dist_{H}^{(J,K)}(\bx,\bw) \geq \ell(\bx) - \ell_U \, ,
\end{equation}
by the triangle inequality and \eqref{eq:mindistD}.
As a consequence, \eqref{eq:expboundNI} and Lemma~\ref{lem:sum}
 yield
\begin{align}
& \widehat{\mathbb{E}}_I[ |\Delta|^{s} ] \; \left(\frac{|\lambda|}{ \big\| \mathcal{U}(\balpha) \big\|_\infty^\beta}\right)^{s}  \exp\left\{ \frac{\beta}{2\xi}\ell(\bx) \right\} \notag \\
	& \qquad\leq  A 
	\sum_{\bw \in \mathbb{Z}^{nd}}  \exp\left\{ -\frac{ \beta}{2\xi} \dist_{H}^{(J,K)}(\bx,\bw)\right\}  < \infty \, . 
\end{align}
The proof is concluded by noting that $ \big\| \mathcal{U}(\balpha) \big\|_\infty^\beta \leq C  \, |\lambda | $ for some $ C < \infty $ for all 
$ (\lambda,\balpha) \in \Gamma^{(p)}_{n-1} $ since the latter is sub-conical.
\end{proof}


The above result will allow us to insert in sums over  $n$-particle configurations  a restriction to ones of  a limited  diameter.
The following bound will be useful for estimates of the remainder.

\begin{corollary}\label{cor:apriori}
Under the hypothesis of Theorem~\ref{thm:apriori}, there exist $ s \in (0,1) $ and $ A , \xi < \infty $ such that 
 for all  $0<  r' \leq r $, the quantity
\be \label{eq:apriorikn1}
	  k_s(\Omega,r,r') \ := \  \sup_{ |x-y |\geq 2r}\mkern-10mu\sum_{\substack{\by \in\mathcal{C}^{(n)}_r(\Omega;y) \\
	  \bx \in\mathcal{C}^{(n)}(\Omega;x) \backslash\mathcal{C}^{(n)}_{r'}(\Omega;x) }} \mkern-10mu\mathbb{E}\left[ \left| G_\Omega(\bx,\by;z) \right|^{s} \right]   
\ee
satisfies  for all $(\lambda,\balpha) \in \Gamma^{(p)}_{n-1}$:
\be \label{eq:apriorikn2}
	  k_s(\Omega,r,r') \ \leq \ 
	  A \,  r^{d(n-1)}  \, |\Omega|^{n-1} \exp\left[-\frac{r'}{(n-1) \, \xi }  \right]  \notag  \, .
\ee
\end{corollary}
\begin{proof}
By Lemma~\ref{lem:extension}, for all $\bx,\by$ 
in the sum in \eqref{eq:apriorikn1}: 
$	 \dist_{\mathcal H}(\bx,\by) \geq  |x-y| -  r  \geq r \geq r'/(n-1)$.
Theorem~\ref{thm:apriori} 
hence guarantees that for some $ s \in (0,1) $ and $ A , \xi < \infty $:
\begin{equation}
	 \mathbb{E}\left[ \left| G_\Omega(\bx,\by;z) \right|^{s} \right] \leq A \exp\left[- \frac{r'}{(n-1)\xi }\right] \, . 
\end{equation}
 The proof is completed by bounding the number of configurations in $ \mathcal{C}^{(n)}_r(\Omega;y) $ and $ \mathcal{C}^{(n)}(\Omega;x) $ by 
$  \left|\mathcal{C}_r (\Omega;y) \right| \leq {n} \, (4r)^{d({n}-1)} $ and $   \left|\mathcal{C} (\Omega;x) \right| \leq  n\, |\Omega|^{n-1} $, respectively.
\end{proof}


\section {Proof of the main result}\label{sec:proof}  

We now turn to the proof of Theorem~\ref{thm:main}.  That is, 
we shall prove that there  exists a monotone
sequence of decreasing sub-conical non-empty domains in the parameter space, 
\be
  \Gamma^{(p)}_1  \supseteq \dots  \supseteq \Gamma^{(p)}_n \supseteq \dots  \, ,
  \ee
   such that for 
each $n$, the Hamiltonian \eqref{Ham} with parameters in 
$\Gamma^{(p)}_n$ exhibits 
uniform $ n $-particle localization in the sense of Definition~\ref{def:loc}; 
each of  the domains including regimes of 
large disorder and of weak interaction.  
The proof will proceed by induction on $n$, the induction 
step consisting of a constructive restriction of the domain.

Establishing a sub-conical domain $ \Gamma^{(p)}_n $ of uniform $ n $-particle localization would be sufficient for our purpose, since Lemma~\ref{cor:efcor}  
then implies that there is $ A , \xi \in (0,\infty) $ such that for all $ (\lambda, \balpha) \in \Gamma^{(p)}_n $, all $ k \in \{ 1 , \dots, n \} $, and all $ \bx , \by \in (\Z^d)^k $:
\be 
\sup_{\Omega \subset \Z^d }\, Q^{(k)}_\Omega(\bx,\by;\R) \ \leq \ A \, e^{-\dist_{\mathcal H}(\bx,\by)/\xi} \, . 
\ee 
Spectral and dynamical localization for $ H^{(k)} $,  as claimed in~1.\ and 2.\ of Theorem~\ref{thm:main}, follow using Theorem~\ref{thm:corr->spec}.  
Assertion~3.\ on the shape of the nested decreasing domains of uniform localization will be verified in the course of  the inductive construction.  \\


\subsection{Analyzing clustered configurations} 

An essential component of the proof is to show the exponential decay of the finite volume quantities: 
\begin{equation}\label{eq:defB}
 B^{(n)}_s(L) := \sup_{\substack{ I \subset \mathbb{R} \\ |I | \geq 1 }}
  \sup_{\Omega \subseteq \Lambda_L} \,  |\partial \Lambda_L|\,  \sum_{y\in \partial \Lambda_L} \sum_{\substack{\bx \in\mathcal{C}^{(n)}_{r_L}(\Omega;0)\\ \by \in\mathcal{C}^{(n)}_{r_L}(\Omega;y)} } 
 	\widehat{\mathbb{E}}_I\left[ \left|  G^{(n)}_\Omega(\bx,\by) \right|^{s} \right] \, , 
\end{equation}
where $ \Lambda_L := [-L,L]^d \cap \mathbb{Z}^d $, and $ s \in (0,1) $.  
One may note that the sum is resticted to configurations in the form of separate 
clouds of particles with diameter less than $ r_L := L/2  $,  which are guaranteed to include at least a pair of sites at distance $L \in \mathbb{N}$ apart. 
By the Wegner estimate \eqref{eq:frac}, for all $ s \in (0,1) $ and $ L  \in \mathbb{N} $:
 \begin{equation}\label{eq:Bbound}
 	B^{(n)}_s(L) \leq \frac{C\, n^2}{|\lambda|^s} \,  L^{2d(n -1)} \, ,
 \end{equation}
 where the constant $ C = C(s,d) < \infty $ is independet of $ (\lambda,\balpha ) \in \R^n $.

The following rescaling principle will be used to show that  
$ B^{(n)}_s(L) $ decays exponentially provided that there is some $ L_0 $ for which it is sufficiently small.   
In its formulation, we consider length scales which grow as  
\begin{equation}
 L_{k+1} := 2 (L_{k} + 1)  \, ,
\end{equation}
i.e., $ L_{k} = 2^k (L_0 +2) -2 $ for all $ k \in \mathbb{N}_0 $.

\begin{theorem}\label{thm:renorm}
 Let $ \Gamma^{(p)}_{n-1} \subset \R_+ \times \R^{p-1} $  be  a sub-conical  domain of uniform $ (n-1) $-particle localization, with $ n \geq 2 $.  
Then there exists $ s \in (0,1) $, $ a, A  , p < \infty $, and $ \nu > 0 $ such that 
\begin{equation}\label{eq:renormineq}
 B^{(n)}_s(L_{k+1}) \leq \frac{a}{|\lambda|^{s} }\, B^{(n)}_s(L_{k})^2 + A  \, L_{k+1}^{2p} \, e^{- 2\nu L_{k}} 
\end{equation}
for all $(\lambda,\balpha) \in \Gamma^{(p)}_{n-1} $ and all $ k \in \mathbb{N}_0 $.
\end{theorem}
In order to keep the flow of the main argument clear, we postpone the proof of this assertion  to Subsection~\ref{subsec:renorm}.    

Quantities satisfying rescaling inequalities as in \eqref{eq:renormineq} are exponentially decreasing provided that they are small on some scale.
This is the content of the following lemma.
For its application we note that 
\begin{equation}
S(\tilde{L}_k) :=B^{(n)}_s(L_k) \, , \qquad \tilde{L}_k := 2^k (L_0+2 ) \, , 
\end{equation}
satisfy \eqref{eq:recursionS}.

\begin{lemma}\label{lem:recS}
Let $  S(L) $ be a non-negative sequence satisfying: 
\begin{equation}\label{eq:recursionS}
	 S(2L) \leq a \, S(L)^2 + b\, L^{2p} \,  e^{- 2 \nu L} 
\end{equation}
for some $ a , b, p \in [0,\infty) $ and $ \nu > 0 $.  If for some $ L_0 > 0 $ there exists $ \eta < \infty $ such that 
\begin{enumerate}
	\item $ \displaystyle \eta^2 \geq a b+ \eta \; \frac{2^p}{L_0^p}\, $,
	\item $ \displaystyle 1 > a \, S(L_0) + \eta \, L_0^p \, e^{-\nu L_0} \; =: e^{-\mu L_0} \, $,
\end{enumerate} 
then for all $ k \in \mathbb{N}_0 $:
 \begin{equation} 
 	S(2^k L_0) \leq a^{-1} \, \exp\left( - \mu\, 2^k L_0 \right) \, . 
  \end{equation} 
  \end{lemma}
\begin{proof}
From \eqref{eq:recursionS} it follows that the quantity $ R(L) := a S(L) + \eta L^p e^{-\nu L} $ satisfies 
\begin{align}
	R(2L) & \leq \left( a S(L) \right)^2 + \left( a b +  \eta \frac{2^p}{L_0^p}\right) L^{2p}  e^{- 2 \nu L}  \notag \\
		& \leq  \left( a S(L) \right)^2 + \eta^2 L^{2p}  e^{- 2 \nu L} \leq R(L)^2 \, .
\end{align}
The claimed \eqref{eq:recursionS} follows by iteration, using  $ R(L_0) = e^{-\mu L_0 }  $.
\end{proof}

\subsection{The inductive proof}
\begin{proof}[Proof of Theorem~\ref{thm:main}]  
As explained above in this section, it suffices to establish sub-conical domains of 
uniform $ n $-particle localization (in the sense of Definition~\ref{def:loc}).  

As an induction anchor 
we use the fact \cite{AM} that there is some $ \lambda_1 \in (0,\infty) $ such that $ \Gamma^{(p)}_1 := \big(\lambda_1,\infty\big) \times \R^{p-1} $ serves as a domain of uniform localization
for the one-particle Hamiltonian $ H^{(1)} $. (Note that the last $ (p-1) $-components of $\Gamma^{(p)}_1 $ are irrelevant for $ H^{(1)} $.)

In the induction step ($n-1 \to n $), we  will first construct a robust, sub-conical domain $ \Gamma^{(p)}_n \subseteq  \Gamma^{(p)}_{n-1}  $ and pick $ L_0 \in \mathbb{N}$ sufficiently large such that at some inverse localization length $ \mu > 0 $:
\begin{equation}\label{eq:expdecayBs}
	 B^{(n)}_s(L_k) \leq  \frac{|\lambda|^{s}}{a} \,  e^{-\mu (L_k + 2)} \leq 2 B^{(n)}_s(L_0) \;  e^{-\mu (L_k - 2 L_0)} \, ,
\end{equation}
for all $ k \in \mathbb{N}_0 $ and all $(\lambda,\balpha) \in \Gamma^{(p)}_n $.
Based on the induction hypothesis, Theorem~\ref{thm:renorm} and Lemma~\ref{lem:recS}, this is will done separately in the two regimes of interest:
\begin{description}
\item[Strong disorder regime:]
Here we choose $ L_0 \in \mathbb{N} $ by the condition $ 2^p e^{-\nu L_0} < 1/4 $. 
Based on that we pick $ \eta $ in the range 
\begin{equation}
	\frac{2^{p+1}}{L_0^p}  < \eta < \frac{1}{2 L_0^p} \, e^{\nu L_0} \, ,  
\end{equation} 
which is non-empty by our previous choice of $ L_0 $. Notice that for this choice $ \eta \, 2^p / L_0^p \leq \eta^2 /2 $.
We now restrict the domain $ \Gamma^{(p)}_{n-1} $ by choosing $ \lambda_n \geq \lambda_{n-1} $ large enough such that for all $ \lambda \geq  \lambda_n $:
\begin{equation}
	\frac{a\, A }{\lambda^s} < \frac{\eta^2}{2} \, \qquad \mbox{and} \qquad \frac{a}{\lambda^s} \; B^{(n)}_s(L_0 ) < \frac{1}{2} \, .
\end{equation}
This is possible by~\eqref{eq:Bbound} and the fact that $ A  $ is independent of $ \lambda $.
Lemma~\ref{lem:recS} hence guarantees that 
\begin{equation}\label{eq:mu}
	\mu = - \frac{\log(a\,  \lambda_n^{-s} B^{(n)}_s(L_0 ) + 1/2 ) }{L_0 } \geq - \frac{\log(2 \, a\, \lambda_n^{-s} B^{(n)}_s(L_0 )) }{2 L_0} > 0 
\end{equation}
serves as an inverse localization length in \eqref{eq:expdecayBs}, which is valid for all $ (\lambda, \balpha) $ in the following regime of strong disorder:
\be \label{eq:highdisreg}
 D_n^{(p)} :=\left\{ (\lambda, \balpha)  \, \big| \, \lambda > \lambda_n \right\} \cap \Gamma_{n-1}^{(p)} \, . 
\ee
\item[Weak interaction regime:]
The localization bounds in \cite{AM} and arguments as in Theorem~\ref{thm:noninteracting} ensure that $ B^{(n)}_s(L) \to 0  $ as $ L \to \infty $ for $ \balpha =  \boldsymbol{0} $ and all $\lambda  > \lambda_1 $.
We then pick $ L_0 \in \mathbb{N} $ and $ \eta $ such that for $  \balpha = \boldsymbol{0} $:
\begin{align}
	& \frac{a}{\lambda_1^s} B^{(n)}_s(L_0) + \eta L_0^p e^{-\nu L_0} < 1 \, , \label{eq:wd1} \\
	& \frac{a \, A}{\lambda_1^s}  + \eta \frac{2^p}{L_0^p} \leq \eta^2 ,  \label{eq:wd2}
\end{align} 
Since the finite-volume quantity $ B^{(n)}_s(L_0) $ is continuous in $ \balpha $ at $\boldsymbol{0} $, for every $ \lambda > \lambda_1 $ there exist $ \balpha(\lambda)_j > 0 $,
$ j \in \{ 1 , \dots , p \} $ such that \eqref{eq:wd1} and  \eqref{eq:wd2} is maintained for
all $ (\lambda,\balpha) $ in the  following regime of weak interaction
\be \label{eq:weakintreg}
I_n^{(p)} := \left\{ (\lambda, \balpha)  \, \big| \, \lambda > \lambda_1 \; \mbox{and} \; |\balpha_j | < \balpha(\lambda)_j  \; \mbox{for all $ j \in \{ 1, \dots , p \} $}\right\} \cap \Gamma_{n-1}^{(p)} \, . 
\ee
By Lemma~\ref{lem:recS}, $ \mu $ as in \eqref{eq:mu} with $ \lambda_n  $ replaced by $ \lambda_1 $ hence 
serves as an inverse localization length in \eqref{eq:expdecayBs}, which holds for all $(\lambda,\balpha) \in I_n^{(p)}$.
\end{description}
Summarizing, we have thus 
established \eqref{eq:expdecayBs} for all $ (\lambda , \balpha) \in D_n^{(p)} \cup I_n^{(p)} $ and a suitably large $ L_0 \in \mathbb{N} $. By further restriction, we may find thus find a robust, but
sub-conical domain
\be 
	\left( \emptyset \neq \right) \; \Gamma^{(p)}_n \subseteq D_n^{(p)} \cup I_n^{(p)} 
\ee
for which \eqref{eq:expdecayBs} holds. \\

To complete the induction step, it requires to establish the exponential decay \eqref{eq:exp} of the $ n $-particle Green function for all $ (\lambda, \balpha) \in \Gamma^{(p)}_n $.

For this purpose,
we select $ x \in X $ and $ y \in Y $ such that $ \dist_{\mathcal{H}}(\bx,\by) = |x-y| $.
In case $ |x-y| > L_0 $, which we may assume without loss of generality (wlog), there exists a unique $ k \in \mathbb{N}_0 $ such that the box $ \Lambda_{L_k}(x) $ which is centered at $ x $ satisfies:
\begin{itemize}
\item
$	y \not\in \Lambda_{L_k}(x) $ and $ y \in \Lambda_{L_{k+1}}(x)$,
\item $ L_k \leq | x - y | \leq c\, L_k $ for some $ c  \in (0,\infty) $.
\end{itemize}
We may furthermore assume wlog  that $ \diam \bx < L_k /2 $ is sufficiently small such that $ \bx \in \mathcal{C}^{(n)}(\Omega \cap \Lambda_{L_k}(x)) $, since otherwise \eqref{eq:exp} (with $ k = n $) follows from Theorem~\ref{thm:apriori}.
The resolvent identity, in which we remove all the terms in the Laplacian which connect 
$ \Omega \cap \Lambda_{L_k}(x) $ and its complement, then implies
\begin{align}\label{eq:container}
	\widehat{\mathbb{E}}_I\left[ \left| G_\Omega(\bx,\by)\right|^s\right]  \leq & \sum_{u \in \partial \Lambda_{L_k}(x)} \sum_{\substack{ \bw' \in \mathcal{C}^{(n)}(\Omega;\Omega \backslash  \Lambda_{L_k}(x)) \\ \bw \in \mathcal{C}^{(n)}(\Omega \cap \Lambda_{L_k};u)}} 
		\mkern-10mu  |H_{\bw,\bw'}|^s \notag \\
			& \mkern100mu \times \widehat{\mathbb{E}}_I\left[  \left|G_{\Omega \cap \Lambda_{L_k}(x)}(\bx,\bw)\right|^s \, \left|G_\Omega(\bw',\by) \right|^s\right]  \notag \\
		 \leq & \; \frac{C}{|\lambda|^s} \sup_{\widetilde\Omega \, : \,  \widetilde\Omega \subset \Lambda_{L_k}(x)} \sum_{u \in \partial \Lambda_{L_k}(x)} \sum_{ \bw \in \mathcal{C}^{(n)}(\widetilde\Omega;u)} 
		   \widehat{\mathbb{E}}_I\left[  \left|G_{\widetilde\Omega}(\bx,\bw)\right|^s\right] \, ,
\end{align}
where $ H_{\bw,\bw'} $ are the matrix elements of the Hamiltonian between $ \delta_\bw $ and $ \delta_{\bw'} $.
Inequality~\eqref{eq:container} follows from \eqref{eq:frac} by first conditioning on all random variables apart from those associated with sites $ y $ and $ u' $ which are both outside $ \Lambda_{L_k}(x) $.
The  sum in \eqref{eq:container} is split into two parts, depending on the diameter of the configuration $ \bw $:
\begin{multline}
\widehat{\mathbb{E}}_I\left[ \left| G^{(n)}_\Omega(\bx,\by)\right|^s\right]  \leq  \frac{C}{|\lambda|^s} \, B^{(n)}_s(L_{k}) \\
		+ \sup_{\widetilde \Omega \, : \, \widetilde \Omega \subset \Lambda_{L_k}(x)}  \frac{C}{|\lambda|^s} \, \sum_{u \in \partial \Lambda_{L_k}(x)} 
		\sum_{ \bw \in \mathcal{C}^{(n)}_{r_{L_k}}(\widetilde \Omega;u) }
	 \widehat{\mathbb{E}}_I\left[  \left|G^{(n)}_{\widetilde \Omega}(\bx,\bw)\right|^s\right] \, .
\end{multline}
By \eqref{eq:expdecayBs} both, in the strong disorder regime and the regime of weak interactions, the first term on the right side is exponentially decaying in $ L_k $. To bound the second term we employ Theorem~\ref{thm:apriori} again,
\begin{equation}
	\widehat{\mathbb{E}}_I\left[  \left|G^{(n)}_{\Omega}(\bx,\bw)\right|^s\right] \leq  A \, \exp\left\{-\frac{L_k}{2 (n-1)\xi } \right\} \, ,
\end{equation}
since $ \dist_{\mathcal{H}}(\bx,\bw) \geq |x-u| - \diam(\bx) \geq  L_k /2$ and $ \ell(\bw) \geq L_k/ (2 (n-1)) $.
The sums in the second term have at most $|\partial  \Lambda_{L_k}(x)| $ and $ n | \Lambda_{L_k}(x)|^{n-1} $ terms, respectively. 
Hence, this second term is also exponentially decaying in $ L_k $ and hence in the Hausdorff distance, $ \dist_{\mathcal{H}}(\bx,\by) $.
This completes the proof of the induction step.
\end{proof}


\subsection{Proof of the rescaling inequality}\label{subsec:renorm}

In the above proof we have postponed the derivation of  Theorem~\ref{thm:renorm}, which provides an essential step for establishing regimes of exponential decay.  We shall do that now. 
The construction and the analysis are inspired by arguments familiar from the one-particle setup. 
The following lemma essentially extends arguments in \cite{ASFH}.

\begin{lemma}\label{lem:ren}
Let $ \Omega \subset \mathbb{Z}^d $ and $ V, W \subset \Omega $ with 
$ \dist(V,W) \geq 2 $. Then for all  $ \bx \in\mathcal{C}^{(n)}(V) $ and $ \by \in\mathcal{C}^{(n)}(W) $ 
\begin{multline}\label{eq:ren}
\mathbb{E}\left[ \left| G^{(n)}_\Omega(\bx,\by;z) \right|^s \right]  \leq \frac{C}{|\lambda|^s} \sum_{u \in \partial V \backslash \partial \Omega} \sum_{\bw \in\mathcal{C}^{(n)}(V;u)} \mathbb{E}\left[ \left| G^{(n)}_{V}(\bx,\bw;z) \right|^s \right] \\
\times \sum_{v \in \partial W \backslash \partial \Omega} \sum_{\bv \in\mathcal{C}^{(n)}(W;v)} \mathbb{E}\left[ \left| G^{(n)}_{W}(\bv,\by;z) \right|^s \right] \, ,
\end{multline}
where the constant $ C=C(s,d) < \infty $ is independent of $ (\lambda,\balpha) \in \R^p  $.
\end{lemma}
\begin{proof}
Let $ H_{\bw,\bw'} $ denote the matrix element of $ H_\Omega $ between $ \delta_\bw $ and $ \delta_{\bw'} $.
A twofold application of the resolvent identity, in which we remove all terms in the Laplacian entering $ H_\Omega $ which connect sites in $ V $ with $ \Omega \backslash V $ and $ W $ with $\Omega \backslash W $ , yields
\begin{equation}\label{eq:stopping}
 G_\Omega(\bx,\by;z) = \mkern-40mu\sum_{ \substack{\bw \in \mathcal{C}^{(n)}(V) , \bw' \in  \mathcal{C}^{(n)}(\Omega;\Omega\backslash V)  \\ \bv \in \mathcal{C}^{(n)}(W) , \bv' \in  \mathcal{C}^{(n)}(\Omega;\Omega\backslash W) }} \mkern-50mu G_{V}(\bx,\bw;z) \, H_{\bw,\bw'}\,  G_\Omega(\bw',\bv';z) \,  H_{\bv,\bv'}\, G_{W}(\bv,\by;z) \, .
\end{equation} 
Using $ |a+b|^s \leq |a|^s + |b|^s $ for any $ s \in (0,1) $, we proceed by establishing the following implication of Theorem~\ref{prop:frac},
\begin{multline} \label{eq:factorization}
\mathbb{E}\left[ \left|   G_{V}(\bx,\bw;z) \right|^s \,  \left| G_\Omega(\bw',\bv';z)\right|^s \, \left| G_{W}(\bv,\by;z) \right|^s \right] \\
\leq   \frac{C}{|\lambda|^s} \, \mathbb{E}\left[ \left|   G_{V}(\bx,\bw;z) \right|^s \right] \, \mathbb{E}\left[  \left| G_{W}(\bv,\by;z) \right|^s \right]  \, . 
\end{multline}
For its proof we note that there are two sites $ u',v' \not\in V \cup W $ for which the configurations $ \bw' $ and $ \bv' $
in the sum in \eqref{eq:stopping} have a particle at $ u' $ respectively $ v' $. We may therefore first condition on all random variables $ \left\{ V(x) \right\}_{x\not\in \{u',v'\}} $ and use \eqref{eq:frac}
to estimate the factor in the middle. The other factors are independent of each other and of $ V(u'), V(v') $, such that the expectation value factorizes.
 
The proof of  \eqref{eq:ren} is completed by noting that for any fixed $ \bw $ the number of $ \bw'  \in  \mathcal{C}^{(n)}(\Omega;\Omega\backslash V) $ with $ H_{\bw,\bw'} \neq 0 $  is at most~$ 2d $. 
\end{proof}

We are now ready to give a 

\begin{proof}[Proof of Theorem~\ref{thm:renorm}]
We first restrict the sum in the definition of $ B^{(n)}_s(L_{k+1}) $ to configurations $ \bx , \by $ which have a diameter less than $ r_{L_{k}}$,
\begin{equation}\label{eq:renorm}
 \widetilde B^{(n)}_s(L_{k+1}) :=  |\partial \Lambda_{L_{k+1}}|\, \sup_{\substack{ I \subset \mathbb{R} \\ |I | \geq 1 }}\,  \sup_{\Omega \, : \, \Omega \subseteq \Lambda_{L_{k+1}}}\sum_{y\in \partial \Lambda_{L_{k+1}}} 
 \sum_{\substack{\bx \in\mathcal{C}^{(n)}_{ r_{L_{k}}}(\Omega;0)\\ \by \in\mathcal{C}^{(n)}_{ r_{L_{k}}}(\Omega;y)} }  \widehat{\mathbb{E}}_I\left[ \left|  G^{(n)}_\Omega(\bx,\by) \right|^{s} \right] \, .
\end{equation}
The error is controlled with the help of Corollary~\ref{cor:apriori}: since $ y \in \partial \Lambda_{L_{k+1}} $ we have $ |y | \geq L_{k+1} $ and hence
\begin{align}
 	  B^{(n)}_s(L_{k+1}) - \widetilde B^{(n)}_s(L_{k+1}) & \leq 2 \, \left| \partial\Lambda_{L_{k+1}} \right|^2 \, k_s(\Lambda_{L_{k+1}},r_{L_{k+1}}, r_{L_{k}}) 	\notag \\
	  & \leq  A \,  L_{k+1}^{2dn-2} \, \exp\left( - \frac{L_k}{2(n-1) \xi} \right) \, .  
\end{align}
The sum $ \widetilde B^{(n)}_s(L_{k+1}) $ is now estimated with the help of Lemma~\ref{lem:ren}, in which we pick $ V = \Lambda_{L_{k}} \cap \, \Omega $ and $ W = \Lambda_{L_{k}}(y)  \cap \, \Omega$, where 
the last box is centered at $ y  \in \partial\Lambda_{L_{k+1}}$. 
Since $ \mathcal{C}_{r_{L_{k}}}^{(n)}(\Omega;0) \subseteq  \mathcal{C}^{(n)}_{r_{L_{k}}}(V;0) $ and $ \mathcal{C}_{r_{L_{k}}}^{(n)}(\Omega;y) \subseteq  \mathcal{C}^{(n)}_{r_{L_{k}}}(W;y) $
we thus obtain
\begin{align}\label{eq:sumsB}
	\widetilde B^{(n)}_s(L_{k+1})  \leq &\;  \frac{C}{|\lambda|^{s}}  \;  \left| \partial \Lambda_{L_{k+1}}\right|^2 \,  \sup_{\substack{ I \subset \mathbb{R} \\ |I | \geq 1 }} \Bigg\{ 
	\sup_{\Omega \, : \, \Omega \subseteq \Lambda_{L_{k}}} \sum_{u \in \partial \Lambda_{L_k}}
	\sum_{\substack{\bx \in \mathcal{C}^{(n)}_{r_{L_{k}}}(\Omega;0) \\ \bw \in\mathcal{C}^{(n)}(\Omega;u)}} \widehat{\mathbb{E}}_I\left[ \left| G^{(n)}_{\Omega}(\bx,\bw) \right|^{s} \right] \notag \\
	& \qquad \times \sup_{y \in \partial \Lambda_{L_{k+1}}} \sup_{\Omega \, : \, \Omega \subseteq \Lambda_{L_{k}}(y)}   \sum_{v \in \partial \Lambda_{L_k}(y)}\sum_{\substack{\by \in \mathcal{C}^{(n)}_{r_{L_{k}}}(\Omega;y) \\ \bv \in \mathcal{C}^{(n)}(\Omega;v)}}  \widehat{\mathbb{E}}_I\left[ \left| G^{(n)}_{\Omega}(\bv,\by) \right|^{s} \right] \Bigg\} \, . 
\end{align}
Thanks to translation invariance we may shift $ y $ to the origin  in the last line. 
Moreover, $  \left| \partial \Lambda_{L_{k+1}}\right| \leq  4^{d-1}  \left| \partial \Lambda_{L_{k}}\right| \, $ and we may again restrict the summation to
 configurations with a smaller diameter using Corollary~\ref{cor:apriori} again, 
\begin{align}
\widetilde B^{(n)}_s(L_{k+1})  & \leq  4^{2d-2} \frac{C }{|\lambda|^{s}} \left[ B^{(n)}_s(L_{k}) +  \left| \partial \Lambda_{L_{k}}\right|^2  k_s(\Lambda_{L_{k}}, r_{L_k}, r_{L_k}) \right]^2 \notag \\
	& \leq 2^{4d-3}  \frac{C }{|\lambda|^{s}}  \left[ B^{(n)}_s(L_{k})^2 +  A^{2} \,  L_{k}^{2dn-2} \exp\left( - \frac{L_k}{(n-1) \xi} \right) \right] \, . 
\end{align}
\end{proof}


\bigskip
\bigskip
\appendix 
\addcontentsline{toc}{section}{Appendix}

\noindent {\LARGE \bf Appendix}

\section{Some distances and separation lemmata} \label{app:dist}

Following are some natural  lengths associated with $n$-particle configurations, and some elementary geometric estimates which are of use within this work.     

In general, for a configuration which is denoted by a bold lower-case letter, we shall use the corresponding capital letter to denote its {\rm footprint} in $\Z^d$, which is the subset which it covers; e.g., for $\bx = \{x_1, ..., x_n\}\subset (\Z^d)^n$  we let $X=\bigcup_{i=1}^{n} \{ x_i\} \subset \Z^d$.  Also, for subsets   the index set, $J\subset \{1,..,n\}$, we let $X_J:=\bigcup_{i\in J} \{ x_i\} \subset(\Z^d)$.  

\subsection{Splitting width}

Two convenient measures of the spread of a configuration 
$\bx = \{x_1,...,x_n\}$, are:
\begin{enumerate} 
\item the {\em diameter}, $\rm{diam}(\bx):=\max_{j,k\in \{1,..., n\} }  |x_j-x_k| $
\item the (maximal) {\em  splitting width}, which we define as the supremum over $r$ for which 
there exists a partition of the set $X$ into two subsets at distance $r $ apart, or: 
\begin{equation} \label{eq:defsplit}
\ell (\bx) \ := \ \max_{ J,\,K:\, \ J\cup K = \{1,..,n \} } \rm{dist} (X_J,\, X_K)  \, , 
\end{equation} 
where $\rm{dist}(A,B) = \min_{u\in A, \, v\in B} |u-v|$, with $|u-v|$ the Euclidean distance. 
\end{enumerate}  
Since both quantities depend only on the footprint of $\bx$, in a harmless abuse of notation we may also refer to  $\rm{diam}(\bx)$  as $\rm{diam}(X) $ and to $\ell(\bx)$ as $\ell(X)$.

\begin{lemma} \label{lem:R}
For any configuration $\bx \in (\Z^d)^n$:  
\begin{equation} \label{eq:R3}
 \frac{1 }{ n-1 } \rm{diam}(\bx)  \  \le \  \ell(\bx) \ \le \   \rm{diam}(\bx) \, . 
\end{equation} 
\end{lemma}
\begin{proof}    
 The upper bound is totally elementary.   To prove the lower bound on $\ell(\bx)$ consider the one-parameter family of sets 
\begin{equation} 
X_r  \ := \  \bigcup_{j=1}^{n} \{y\in \Z^d \, : \, |y-x_j| \le r \} \, .
\end{equation} 
For any  $r>0$ such that $X_r$ is connected one clearly has 
\begin{equation} \max_{j,k\in \{1,...,n\}  } |x_j-x_k|  \ \le \ 2 r (n-1) \, . 
\end{equation} 
It follows that for any $r$ such that   $\, \,  2r < \rm{diam}(\bx) /(n-1)  \,  $ the set  $X_r$ is not connected, and hence there is a partition of the configuration $\bx$ into two subsets whose points are at distances greater than $2r$.  This implies  that also $2r \le \ell(\bx) $.  Optimizing over such $r$ we find that $\rm{diam}(\bx)/(n-1)  \le \ell(\bx)$.  
 \end{proof} 
 
 \subsection{Distances in the configuration space}
 
In addition to the regular distance between subsets of $\Z^d$ which is mentioned above, there exists also the notion of the Hausdorff distance, which  is defined as: 
\begin{equation}  
\rm{dist}_{\mathcal H} (X,\, Y)   \ : = \ \max \{ 
\max_{u \in X } \rm{dist} ( \{u \} , Y ),\ 
 \max_{v \in Y } \rm{dist} (\{v\} , X) \}  \, , 
\end{equation}
for any $X,\, Y \in \Z^d$.

In a slight abuse of notation we shall employ the symbol $\rm{dist}_{\mathcal H}$ also for the induced Hausdorff distance between configurations (i.e., $\bx,\, \by\in (\Z^d)^n$): 
\begin{equation} \rm{dist}_{\mathcal H} (\bx,\, \by)\ := \  \rm{dist}_{\mathcal H} (X,\, Y)\ \, . 
\end{equation}
This distance  is clearly most sensitive to the outliers.  Another notion is the symmetrized distance: 
\begin{equation} 
\rm{dist}_S(\bx,\by) \ := \ \min_{\pi \in S_n}  \sum_{j=1}^n |x_j-y_{\pi j}| \, , 
\end{equation} 
with $S_n$ is the permutation group of the  $n$ elements $\{1, ...,n\}$.  \\


The following is an elementary consequence of the definitions.
\begin{lemma}\label{lem:extension}
 Let $ \Omega \subseteq \mathbb{Z}^d $ and $u,v \in \Omega $. For any two configurations $ \bx \in \mathcal{C}^{(n)}(\Omega;u) $ and   $ \by \in \mathcal{C}^{(n)}(\Omega;v) $ , which have a particle at $ u $ and, respectively, $ v  $: 
\begin{equation}
	\dist_{\mathcal H}(\bx,\by) \geq \max\{ \dist(\{u\},Y) , \dist(\{v\},X) \} \geq  |u-v| -  \min\{ \rm{diam}(\bx) , \rm{diam}(\by) \} \, . 
\end{equation}
\end{lemma}
 
For convenience let us also place here the bound:  
\begin{lemma}\label{lem:sum}
 Let $ \Omega \subseteq \mathbb{Z}^d $ and $\bx \in (\Z^d)^n$.
 \begin{enumerate}
 \item  For any site  $u \in \Omega$ with $\dist(\{u\}, X) = L$, and any  $\xi \ge 0$: 
\begin{equation} \label{eq:sum}
	\sum_{ 
	\by  \in \mathcal{C}^{(n)}(\Omega;u)}  
	e^{ - \dist_{\mathcal{H}}(\bx,\by)/ \xi } \ \le \  
	C \max\{L,\, \xi\}^{d(n-1)} \,  e^{-L/\xi}
	 \, . 
\end{equation} 
with a constant $C=C(n) < \infty$. 
\item For any  $\xi \ge 0$:
\begin{equation}
	\sum_{\by  \in \mathcal{C}^{(n)}(\Omega)}e^{ - \dist_{\mathcal{H}}(\bx,\by)/ \xi }
		\leq   C \;  \xi^{nd} \, ,
\end{equation}
for some $ C= C(n,d) < \infty $.
\end{enumerate}
\end{lemma}
\begin{proof}  It is convenient to use the equality:  
\begin{equation} \label{layer}
	\sum_{ 
	\by  \in \mathcal{C}^{(n)}(\Omega;u)  }
	e^{ - \rm{dist}_{\mathcal H}(\bx,\by)/ \xi } \ = \  
	\int_0^\infty dr \frac{e^{-r/\xi} }{\xi}    
	\sum_{\by  \in \mathcal{C}^{(n)}(\Omega;u)}  \indfct[\rm{dist}_{\mathcal H}(\bx,\by) \, \le \, r]  \, .
\end{equation}
The estimate the sum on the right, we note that the configuration $\by$ needs to have one of its points at $u$, and the rest $(n-1)$ points are all within the distance $r$ from $X$.   A simple estimate  yields: 
\begin{equation} \label{slice}
\sum_{\by  \in \mathcal{C}^{(n)}(\Omega;u);  }  \indfct[\rm{dist}_{\mathcal H}(\bx,\by) \, \le \, r] \  
	\le \  \left\{ \begin{array}{ll} 
	0 & \quad r \le L \\[1ex] 
	n^n (2r)^{d (n-1)}  &  \quad r \ge L   
	\end{array}  \right.
	\end{equation} 
where $2$ could also be replaced by $b_d$, which is the  maximal value of $b$ such that any sphere  in $\R^d$ of radius $r$ includes not more than $b r^d$ lattice points (of $\Z^d)$.  Substituting \eqref{slice} in \eqref{layer} one readily obtains the claimed bound~\eqref{eq:sum}.   
The second claim follows from the first by summation over $ u $. 
\end{proof} 

\section{From eigenfunction correlators to dynamical and spectral information}  \label{app:spec}

In Theorem~\ref{thm:corr->spec} we presented a known method \cite{A2} for the derivation of information on the dynamical and spectral properties of an  infinite-volume  operator from 
bounds on the eigenfuction correlators of its restrictions to finite domains.  For convenience, following is an outline of a proof of this result.

We recall that the assumption is that  for a sequence of finite domains $\Omega$ which converge to $\Z^d$, and a fixed interval $I$: 
 \begin{equation}
\E{ Q_{\Omega}^{(n)}(\bx,\by;I)}   \ \le \  A \; 
e^{-  K (\bx,\by)  } \,  . 
\end{equation} 
with some kernel $K (\bx,\by) $  (e.g., $K (\bx,\by) = \dist(\bx,\by)/\xi$).  
  \\
 
  \begin{proof}[Proof of Theorem~\ref{thm:corr->spec}]
{\em i)\/}  Through a trivial extension of the finite volume operators, they can be naturally viewed as  acting in the same  space $\ell^2(\Z^{nd})$ (acting as $0$ on functions supported outside $\Omega^n$). 
Using the Combes-Thomas estimate on the Green function, one may see that 
the operator $H(\omega)$ is the  limit, in the strong resolvent sense, of any sequence of $H_\Omega(\omega)$, as $\Omega \to \Z^d$  (allowing in the process also  arbitrary self-adjoint boundary conditions at the receding boundary).    Thus, a bound like \eqref{loc_fK} but modified through a restriction of $f$ to continuous functions can be deduced from eq.~\eqref{Q-F} and the general properties of strong resolvent convergence
(cf. \cite[Thm.~VIII.20]{ReSi3}).   Since, by the Wegner estimate, the mean density of states is a continuous measure, Lusin's approximation theorem allows to extend the resulting  bound to all measurable and bounded functions, thus yielding \eqref{loc_fK}.  Of  particular interest is the implied dynamical localization bound: 
\begin{equation} \label{loc_dyn}
\E{\sup_{t\in \R}\left|\langle \delta_\bx \, ,  \, P_I(H) \, e^{-itH} \delta_\by \rangle \right|}   \ \le \  A\, 
e^{-  K(\bx,\by) } \,  . 
\end{equation}

{\em ii)\/}  By the RAGE criterion (see, e.g.~\cite{ReSi3})
the projection on the continuous spectrum in the interval $I$ satisfies: 
\begin{align} \label{RAGE}
&\E{ \| P_{I; \rm{cont}}(H)  \, \delta_\bx  \|^2} \  = \ 
\notag \\  
& \ = 
\E{  \lim_{R\to \infty} \lim_{T\to \infty} 
\frac{1}{T}  \int_0^T \sum_{\by : \, \dist{(\bx,\by)} \ge R} 
\left|\langle \delta_\bx \, , \, P_I(H) \, e^{-itH}\,  \delta_\by \rangle \right|^2 dt} \notag  \\  
 & \  \le \ 
  \lim_{R\to \infty}  \sum_{\by : \, \dist{(\bx,\by)} \ge R} 
 \E{  
 \sup_{t\in \R} 
  \left|\langle \delta_\bx \, , \, P_I(H) \, e^{-itH}\,  \delta _\by \rangle \right|  } \    \, , 
\end{align} 
where the last inequality is by Fatou's lemma and the natural bound ($1$) on the summed quantity.  Thus, under the assumption \eqref{loc_assump}, 
\be  
\E{ \| P_{I; \rm{cont}}(H)  \,  \delta_\bx \|^2} \  \le   
  \    \lim_{R\to \infty}    
  \sum_{\by : \, \dist{(\bx,\by)} \ge R} \!\! A\,
 e^{-K(\bx,\by)} \      \, .
\ee  
In case $K(\bx,\by) = 2 \dist_{\mathcal{H}}(\bx,\by)/\xi$,  or any of the other distances discussed here (which are only larger), the above limit vanishes.  
Since $\{\delta_\bx \}_{\bx\in (\Z^d)^n}$ is a spanning collection of vectors, one may conclude that, under the assumption \eqref{loc_assump}, within the 
$n$-particle sector $H(\omega)$ has almost surely  no continuous spectrum in the interval $I$.   

{\em iii)\/}  Under the above assumption, we will construct a complete set of exponentially bounded eigenfunctions which form a subset of  $\{ P_{\{E\}}(H(\omega))\, \delta_{\by} \, | \,  E \in \sigma( H(\omega) )  , \, \by \in (\Z^d)^n \} $.  Clearly functions in this collection are either zero or eigenfunctions of $H(\omega)$.  

For the complete set we chose  functions corresponding to configurations
$ \by \in (\Z^d)^n $ which are \emph{$ E $-representative} in the sense that 
\be
	\langle \delta_\by \, , \, P_{\{E\}}(H(\omega)) \, \delta_{\by} \rangle \geq \frac{\left( 1 + |\by|\right)^{-(nd+1)} }{
	\sum_{\bx\in (\Z^d)^n} \left( 1 + |\bx|\right)^{-(nd+1)}} \, . 
\ee

Note that for any $ E \in \sigma( H(\omega) ) $ there is at least one $ E $-representative 
configuration, since $ \sum_{\by}  \langle \delta_\by \, , \, P_{\{E\}}(H(\omega)) \, \delta_{\by} \rangle \geq 1 $.
We claim that the corresponding collection of normalized eigenfunctions: \be \label{eq:complete set}
\psi_{E,\by} := \frac{P_{\{E\}}(H(\omega))\, \delta_{\by}}{
\| P_{\{E\}}(H(\omega))\, \delta_{\by}\|}, \quad  \mbox{with} \; E \in \sigma( H(\omega) ) \cap I \, , \; \by \in (\Z^d)^n \; \mbox{$E$-representative,} 
\ee
 spans the full subspace of eigenfunctions of $ H(\omega) $ corresponding to eigenvalues in $ I $.
For if not then there exists $ E \in \sigma( H(\omega) ) \cap I 
$ and a normalized function satisfying $ \phi  = P_{\{E\}}(H(\omega) ) \phi $, which is within the orthogonal complement of  the subspace spanned by~\eqref{eq:complete set}.  This would imply the following contradiction: 
\be 
	1 = \langle \phi \, , \, \phi \rangle = \mkern-30mu\sum_{\substack{\by  \in (\Z^d)^n \,  \mbox{\small not} \\ \mbox{ \small $ E $-representative}}} \mkern-30mu 
	\left|\langle \phi \, , P_{\{E\}}(H(\omega))\, \delta_{\by} \rangle \right|^2
		\leq  \mkern-30mu\sum_{\substack{\by  \in (\Z^d)^n \,  \mbox{\small not} \\ \mbox{\small  $ E $-representative}}} \mkern-30mu\langle \delta_\by \, , \, P_{\{E\}}(H(\omega)) \, \delta_{\by}\rangle  < 1 \, . 
\ee

For bounds on the eigenfunctions in~\eqref{eq:complete set}, one may apply  the Wiener criterion which (combined with the Fatou lemma)  yields
\begin{multline} \label{Wiener}
\E{\sum_{E\in \sigma(H)\bigcap I} \left|\langle \delta_\bx \, , \, P_{\{E\}}(H) \, \delta_{\by} \rangle \right|^2}   \ \  \le \ \hfill \\
\  \le \  \lim_{T\to \infty} \E{ \frac{1}{T}  \int_0^T 
\left|\langle \delta_\bx \, , \, P_I(H) \, e^{-itH}\,  \delta_{\by} \rangle \right|^2 dt}  \  
\   \le \  A\, e^{- \dist_{\mathcal{H}}(\bx,\by) / \xi} \,  ,
\end{multline} 
as a consequence of \eqref{loc_dyn} in case $ K(\bx,\by) = 2 \dist_{\mathcal{H}}(\bx,\by) / \xi $. 
Summing the resulting bound we get: 
 \begin{equation} \label{Wiener2}
\E{\sum_{\substack{E\in \sigma(H)\bigcap I\\ \bx,\by\in (\Z^d)^n}}  
\frac{e^{\dist_{\mathcal{H}}(\bx,\by) / \xi}}{(1+ |\by|)^{nd+1}}
\left|\langle \delta_\bx \, , \, P_{\{E\}}(H) \, \delta_{\by} \rangle \right|^2 }   \   <  \ \infty \, .
\end{equation} 
Thus, using the Chebyshev principle, there exists a positive function $A(\omega;n)$ of finite mean such that for all $E\in \sigma(H(\omega)) \cap I$ and $\bx \in (\Z^d)^n$: 
\begin{multline} \label{eq:efbound}
  \left| \psi_{E,\by}(\bx) \right|^2 \; \indfct[ \, \by \in (\Z^d)^n\;  \mbox{is $ E$-representative}]  \\ \  \le   \  
 A(\omega;n) \,  (1 +|\by|)^{2nd+2} \ e^{- \dist_{\mathcal{H}}(\bx,\by) /\xi} \, . 
\end{multline}
Since for any $E\in \sigma (H(\omega))$ and any $y$ in the above collection  the function $ \psi_{E,\by} $ is non-negligble  at $y$ in the sense of~\eqref{eq:nonneg}  this proves the last claim which is made in  Theorem~\ref{thm:corr->spec}.
\end{proof}   

\section{An averaging principle}\label{app:averaging}

In the proof of Theorem~\ref{thm:Q<G} we made use of an averaging principle, which is useful for conditional averages where  the value of the potential at a single site, $u\in \Z^d$, is redrawn  at fixed values of the other (random) parameters.  This provides a generalization of Lemma~\ref{lem:QF1} which is suitable for multiparticle systems.   Following is its derivation.

In the statement of the result, use is made 
of the holomorphic  family of operators  
\begin{equation}\label{eq:defKuz}
	K_u(z) \ = \   \sqrt{N_u} (  H_\Omega  -z)^{-1} \!\sqrt{N_u}  \,   
\end{equation} 
which we take as acting  in the range of $ N_u $ within $\mathcal{H}^{(n)}_\Omega $. 
The operator valued function is  analytic in $ \C\backslash \sigma(H_\Omega) $.   For real $z=E\in \R$ the operators are self adjoint, and  for convenience of the (local) argument which follows, we employ an auxiliary index $\nu$ to label the eigenvalues ($\kappa_\nu(E) $),  counted without multiplicity, and the corresponding projection operators ($ \Pi_{\kappa_\nu}(E) $).  Questions of order do not matter here since we shall always be summing over $\nu$.  

Along $ \R \backslash \sigma(H_\Omega) $   
analyticity implies that for all but possibly finitely many exceptional values of $ E  $, at which level-crossings occur,
both eigenvalues and projections may be analytically continued to a small neighborhood of $ E $ on which the spectral representation 
\begin{equation} \label{eq:repK}
	K_u(z) \ = \  
	\sum_{\nu}  \kappa_\nu(z) \, \Pi_{\kappa_\nu}(z) \, , 
	\end{equation} 
holds.  The operators $\Pi_{\kappa_\nu}(z) $ are projections, satisfying  $ \Pi_{\kappa_\nu}(z) \, \Pi_{\kappa_{\nu'}}(z)= \Pi_{\kappa_\nu}(z) \delta_{\nu,\nu'} $, though they are orthogonal projections only for real $z$; cf.~\cite[Ch.~II]{Kato}.

\begin{lemma} [\, =\, Lemma~\ref{lem:specavQ}\, ]
Let $ s \in [0,1) $ and  $ \Omega \subset \mathbb{Z}^d $ be a finite set, and  $u $ a point in $ \Omega $. Then for all $ \bx \in \mathcal{C}^{(n)}(\Omega;u) $ and $ \by \in \mathcal{C}^{(n)}(\Omega) $:
\begin{align} \label{refresh2}
	& N_u(\bx) ^{1+\frac{s}{2}}\ \, \int_\mathbb{R} Q_\Omega(\bx,\by;I,s) \Big|_{V(x)\to V(x)+v}\, \frac{dv}{|v|^s}  \\
	& = \ \frac{1}{|\lambda|^{1-s} } \int_I \,  \sum_{\nu}
	\langle \delta_\bx ,   \Pi_{\kappa_\nu}(E)   \,  \delta_\bx\rangle^{1-s} \, \left|\langle \delta_\bx  ,  \Pi_{\kappa_\nu}(E)  \sqrt{N_u} \, (  H_\Omega  - E)^{-1}   \delta_\by \rangle\right|^s \, dE \, . \notag
\end{align} 
\end{lemma}
\begin{proof}
Let us consider the family of operators  $\widehat H_\Omega (v)  := H_\Omega + \lambda v N_u $  with the extra parameter $v$, which in effect modifies the value of the potential $V_u$ of $H_\Omega$.  A standard resolvent identity leads to the Krein formula:
\begin{equation} \label{krein}
	\sqrt{N_u} \left(\widehat H_\Omega(v)-z\right)^{-1} = \left( 1 +\lambda v\, K_u(z) \right)^{-1}  \sqrt{N_u} \, ( H_\Omega  -z )^{-1}  \, .  
\end{equation}
 
For each $v$,  the family of operators $\sqrt{N_u} \,  \left(\widehat H_\Omega(v)-z\right)^{-1}$ is an analytic in  $z   \in \C\backslash \sigma (\widehat H_\Omega(v))$,  with residues given by 
$ \sqrt{N_u} \, \widehat P_{\{E\}} $, where  
$\widehat P_{\{E\}} $ are the projection operators on the eigenspaces of 
$\widehat H_\Omega(v)$, at eigenvalues $E\in \sigma (\widehat H_\Omega(v)) $.   
The operator valued function on the right-hand side of \eqref{krein}, is singular if and only if  either: 
\be  \label{spec1}
  -(\lambda v)^{-1} \in \sigma(K_u(E))  
\ee  
or
\be  \label{spec2} 
 E\in \sigma ( H_\Omega) \, .  
\ee  
and we argue next  that the singularities  
at $\sigma(H_\Omega)$ are removable  for almost every value of $v\in \R$.

Since the spectrum of $\widehat H_\Omega(v)$  is monotone non-decreasing in $v$,  its spectral projections can be decomposed  ($\widehat P_{\{E\}} = \widehat  P^{(\rm mon)}_{\{E\}}  +\widehat P^{(\rm fix)}_{\{E\}}$, for $E\in \sigma (\widehat H_\omega (v) )$)  into a part for which the corresponding spectrum is strictly monotone and  another corresponding to spectrum which does not move with $v$.    (Monotonicity plays here only an auxiliary role, and could also be replaced by analyticity in $v$ or just smoothness.)
At Lebesgue almost every $v\in \R$ the spectrum of $\widehat H_\Omega(v)$ corresponding to $\widehat P^{(\rm mon)}_{\{E\}}$ is disjoint from $\sigma (H_\Omega)$, in which case the singularities which  
$(\widehat H_\Omega(v)-z)^{-1}$ may have  at $\sigma (H_\Omega )$ are only due to the fixed part of the spectrum,   with the corresponding residues of 
$\sqrt{N_u} \, (\widehat H_\Omega(v)-z)^{-1}$
being given by $\sqrt{N_u} \,  \widehat P^{(\rm fix)}_{\{E\}}$.
However, functions in the range of the projections $\widehat P^{(\rm fix)} $  are annihilated by $N_u$, and thus  $ \sqrt{N_u} \, \widehat   P^{(\rm fix)}_{\{E\}}  =0$.  Since  $\sqrt{N_u} \, (\widehat H_\Omega(v)-z)^{-1}$ has only simple pole singularities, the vanishing of the residue at $E$ implies that the singularity of the expression on the right-hand side of \eqref{krein} is removable there. 

 In conclusion, for almost every $v\in \R$,  even if there is an overlap in the spectra of   $\widehat H_\Omega(v)$ and $H_\Omega$,  $ \sqrt{N_u} \, \widehat   P_{\{E\}} =0$ for all $E\in \sigma(H_\Omega)$.


We now consider $E\not\in\sigma ( H_\Omega) $ at which \eqref{spec1} holds. Such energies will not coincide with exceptional points of level-crossing for $ K_u(E) $ for almost all $ v \in \R $. 
Therefore
 a simple residue calculation based on \eqref{eq:repK} yields: 
\begin{equation}\label{eq:repproj}
	\sqrt{N_u} \ \widehat  P_{\{E\}} = \frac{\kappa_\nu(E)}{\kappa'_\nu (E)} \,\Pi_{\kappa_\nu}(z)   \sqrt{N_u} \, (  H_\Omega  - E)^{-1}  
\end{equation}
where $\kappa_\nu(E)$ is that eigenvalue of $K_u(E)$ at which  $  -(\lambda v)^{-1}  = \kappa_\nu(E) $ holds, and $\kappa_\nu'(E)$ is the derivative of the eigenvalue with respect to $E$ evaluated at that particular point. 
In particular, 
\be 
\sqrt{N_u} \widehat  P_{\{E\}} \sqrt{N_u} \ = \  \frac{\kappa_\nu(E)^2}{\kappa_\nu'(E)}\,  \Pi_{\kappa_\nu}(E) 
\ = \   -\left[ \frac{d}{d\, E}\,  \kappa_\nu(E)^{-1} \right] \, \Pi_{\kappa_\nu}(E) \, .
\ee 

Using the relation $\delta( \, g(E)\,) \ g'(E) \ =\   \sum_{u\in g^{-1} (\{0\}) } \delta(E-u) $, we find that 
for any $ f $ which is continuous on a neighborhood of $ \sigma(\widehat H_\Omega) $:
\begin{multline}   \label{sum} 
	 \sum_{E \in \sigma(\widehat H_\Omega)\cap I } \langle \delta_\bx ,  \sqrt{N_u}\widehat P_{\{E\}} \sqrt{N_u}  \delta_\bx \rangle\, f(E)   \\ =   \sum_\nu \int_I \delta( \lambda v +  \kappa_\nu(E)^{-1}) \, 
	 \langle \delta_\bx , \Pi_{\kappa_\nu}(E)  \, \delta_\bx \rangle \,  f(E) \, dE \, . 
\end{multline}

The eigenfunction correlator, which is defined as:
\begin{multline}
	N_u(\bx)^{1+\frac{s}{2}}\, Q_\Omega(\bx,\by;I,s)  \Big|_{V(x)\to V(x)+v}\\ =  \sum_{E \in \sigma(\widehat H_\Omega)\cap I } \langle \delta_\bx ,  \sqrt{N_u} \widehat P_{\{E\}} \sqrt{N_u}  \delta_\bx \rangle^{1-s} 
		\left|\langle \delta_\bx ,  \sqrt{N_u} \widehat P_{\{E\}}  \delta_\by \rangle\right|^s \, , 
\end{multline}
can be presented in the form of \eqref{sum}, with $f(E)$ 
the function   which is defined in the neighborhood of the zeros of $\lambda v + \kappa_\nu(E)^{-1} $
as  
\begin{eqnarray} 
f(E) & =&   \left|	\frac{1}{\kappa_\nu(E)} \,  \frac{\langle \delta_\bx ,  \Pi_{\kappa_\nu}(E)   \sqrt{N_u} \, (  H_\Omega  - E)^{-1}   \delta_\by \rangle}{ \langle \delta_\bx ,  \Pi_{\kappa_\nu}(E)   \,  \delta_\bx \rangle} 
\right|^s \nonumber  \\[2ex]  
& =&  \left| \frac{ \langle \delta_\bx ,  \sqrt{N_u} P_{\{E\}}  \delta_\by \rangle}{ \langle \delta_\bx ,  \sqrt{N_u} P_{\{E\}} \sqrt{N_u}  \delta_\bx \rangle} \right|^s \, .
\end{eqnarray}
The claim follows now upon integration over $ v $, of the expression which is obtained by substituting the above function in \eqref{sum}.
  \end{proof}

 \section{On the Wegner estimate} \label{Sect:Wegner}

As noted in Section~\ref{sec:BGF}, the Wegner estimate is not being explicitly used in the Fractional Moment Analysis.  However the statement is of  intrinsic interest, and it provides also a useful tool for various other purposes .  
The basic bound has already been extended to multiparticle systems \cite{ChSu2,Ki07}.   Our purpose here is to comment on the subject from the perspective of the approach used in Section~\ref{sec:BGF}.  

A local version of the bound is  the following statement of finiteness of the conditional mean of the  density of the spectral measure associated with the vector $\delta_{\bx}$ for an arbitrary configuration $\bx \in \mathcal{C}^{(n)}(\Omega) $, obtained by  averaging  over one of the  potential variables associated with  the occupied sites, $ V(x_j; \omega) $, $j\in \{1,...,n\}$. 
In the following, for a self adjoint operator $H$  we denote by $ P_I(H) $ the spectral projection  associated with a Borel set $ I $  whose Lebesgue measure is denoted by $ | I |$.

 \begin{theorem}\label{thm:Wegner1}
Let $u$ be a site in $\Omega  \subseteq \mathbb{Z}^{d}$,  $ \bx \in\mathcal{C}^{(n)}(\Omega;u) $ a configuration with at least one particle at $u$, and 
  \begin{equation}
 	\mu_{\bx}(I; \omega) :=  \left\langle \delta^{(n)}_\bx , P_I(H^{(n)}_\Omega(\omega) ) \,  \delta^{(n)}_\bx \right\rangle \, ,   
\end{equation}
the spectral measure of the operator $ H^{(n)}_\Omega(\omega) $ which is associated with the vector $ \delta_\bx^{(n)}$ and a Borel set $ I \subseteq \mathbb{R}$.    Then the average of $\mu_{\bx}(I; \omega)$  over the values of the potential at $u$ satisfies:   
 \begin{equation}\label{eq:frac1}
 	\int_\R   \mu_{\bx}(I;\omega)
	\varrho(V_u) dV_u\ \equiv \ 
	\mathbb{E}\left( \mu_{\bx}(I,\omega) \, | \, \left\{ V(v) \right\}_{v\neq u} \right) 
	\leq \frac{\|\varrho\|_\infty }{|\lambda|\,  N_u(\bx)} \, |I |  \, . 
\end{equation}
 \end{theorem}
 \begin{proof} 
 Bearing in mind the dependence of the Hamiltonian on $V_u$ as expressed in \eqref{eq:HN}, one may apply the averaging principle (cf.\ \cite{Ko86,SiWo86,CoHi94})  which states that for any self-adjoint operator $ A $ and a bounded operator $ N \geq B^\dagger B $ on a Hilbert space:\begin{equation}\label{eq:specav}
	\int_\mathbb{R}  \left\| B^\dagger \, P_I(A + \tau \, N ) \, B  \right\|   
	\  \varrho(\tau) \, d\tau \leq \|\varrho\|_\infty  \, |I | \, .
\end{equation}
The claim, \eqref{eq:frac1},  follows from \eqref{eq:specav}
 by observing that for 
 $  \bx \in\mathcal{C}^{(n)}(\Omega;u) $
one has: $
 	 N_u(\bx)\, \, \langle \delta_\bx \, ,  P_{I}(H_\Omega) \, \delta_\bx \rangle = \langle   \delta_\bx \, ,  N_u^{1/2} \,  P_{I}(H_\Omega) \, N_u^{1/2} \, \delta_\bx \rangle $.
 \end{proof}

\section*{Acknowledgements}
With pleasure we thank Yuri Suhov for  useful discussions of multiparticle localization,  and Shmuel Fishman and Uzy Smilansky for stimulating discussions of related topics  on a visit to the  Center for Complex Systems at Weizmann Inst. of Science.  Support for the latter  was received from the BSF grant 710021.  The work was supported in parts by the NSF grants DMS-0602360 (MA), DMS-0701181 (SW) and a Sloan Fellowship~(SW).   


\bibliographystyle{plain}

\end{document}